\documentclass{article}

\usepackage{arxiv}

\usepackage[utf8]{inputenc} 
\usepackage[T1]{fontenc}    
\usepackage{algorithmic}
\usepackage{amsmath}
\usepackage{amsthm}
\usepackage{amssymb}
\usepackage{balance}
\usepackage{graphicx}
\usepackage{multirow}
\usepackage{subcaption}
\usepackage{xspace}
\usepackage{url}
\usepackage{weiwAlgorithm}
\usepackage{balance}
\usepackage{makecell}

\newcommand{\stitle}[1]{\vspace{1ex} \noindent{{\bf #1}}}

\newcommand{\sstitle}[1]{\vspace{1ex} \noindent{\textit{ #1}}}

\newcommand{\kw}[1]{{\ensuremath {\mathsf{#1}}}\xspace}

\newcommand{\kwnospace}[1]{{\ensuremath {\mathsf{#1}}}}

\long\def\comment#1{}

\newcommand{\reffig}[1]{Figure~\ref{fig:#1}}
\newcommand{\refsec}[1]{Section~\ref{sec:#1}}
\newcommand{\reftable}[1]{Table~\ref{tab:#1}}
\newcommand{\refalg}[1]{Algorithm~\ref{alg:#1}}
\newcommand{\refeq}[1]{Equation~\ref{eq:#1}}
\newcommand{\refdef}[1]{Definiton~\ref{def:#1}}

\newcommand{\refex}[1]{Example~\ref{ex:#1}}

\newcommand{\ttjoin}{\kwnospace{Twin}\kwnospace{Twig}\kw{Join}}

\newcommand{\starjoin}{\kwnospace{Star}\kw{Join}}
\newcommand{\multiwayjoin}{\kwnospace{Multiway}\kw{Join}}
\newcommand{\cliquejoin}{\kwnospace{Clique}\kw{Join}}
\newcommand{\optjoin}{\kwnospace{BiG}\kw{Join}}
\newcommand{\bigjoin}{\kwnospace{BiG}\kw{Join}}
\newcommand{\crystaljoin}{\kwnospace{Crystal}\kw{Join}}
\newcommand{\psgl}{\kw{PSgL}}

\newcommand{\cflmatch}{\kwnospace{CFL}\kw{Match}}

\newcommand{\genericjoin}{\code{Generic}\code{Join}\xspace}

\newcommand{\er}{\kwnospace{ER}\xspace}
\newcommand{\power}{\kwnospace{PR}\xspace}
\newcommand{\eaat}{\textsc{Bin}\textsc{Join}\xspace}
\newcommand{\vaat}{\textsc{WOpt}\textsc{Join}\xspace}
\newcommand{\other}{\textsc{Others}\xspace}
\newcommand{\rep}{\textsc{Full}\textsc{Rep}\xspace}
\newcommand{\multiway}{\textsc{Shr}\textsc{Cube}\xspace}
\newcommand{\trindex}{\code{Tr}\code{Indexing}\xspace}
\newcommand{\batching}{\code{Batching}\xspace}
\newcommand{\compression}{\code{Compression}\xspace}
\newcommand{\timely}{\code{Timely}\xspace}
\newcommand{\code}[1]{\texttt{#1}}
\newcommand{\ttwig}{\code{Twin}\code{Twig}\xspace}

\newcommand{\Star}{\code{Star}\xspace}

\newcommand{\clique}{\code{Clique}\xspace}

\newcommand{\oom}{\code{OOM}\xspace}
\newcommand{\timeout}{\code{OT}\xspace}

\newcommand{\binaryjoinnotrindex}{\eaat{\tiny (w.o.t.)}\xspace}
\newcommand{\binaryjoinnobatching}{\eaat{\tiny (w.o.b.)}\xspace}
\newcommand{\binaryjoinnocompression}{\eaat{\tiny (w.o.c.)}\xspace}

\newcommand{\genericjoinnotrindex}{\vaat{{\tiny (w.o.t.)}}\xspace}

\theoremstyle{definition}
\newtheorem{example}{Example}[section]
\newtheorem{definition}{Definition}[section]
\newtheorem{theorem}{Theorem}[section]
\newtheorem{remark}{Remark}[section]

\title{A Survey and Experimental Analysis of Distributed Subgraph Matching}

\author{
  Longbin~Lai \\
  School of Computer Science and Engineering \\
  UNSW, Sydney \\
  NSW, 2052 \\
  \texttt{llai@cse.unsw.edu.au} \\
   \And
 Zhu~Qing \\
  School of computer science and software engineering \\
  East China Normal University, China\\
  Shanghai, China \\
  \texttt{zhuqing@stu.ecnu.edu.cn} \\
  \And
  Zhengyi~Yang \\
  School of computer science and software engineering \\
  UNSW, Sydney \\
  NSW, 2052 \\
  \texttt{zyang@cse.unsw.edu.au} \\
  \And
  Xin~Jin \\
  School of computer science and software engineering \\
  East China Normal University, China\\
  Shanghai, China \\
  \texttt{xinjin@stu.ecnu.edu.cn} \\
  \And
  Zhengmin~Lai \\
  School of computer science and software engineering \\
  East China Normal University, China\\
  Shanghai, China \\
  \texttt{zmlai@stu.ecnu.edu.cn} \\
  \And
  Ran~Wang \\
  School of computer science and software engineering \\
  East China Normal University, China\\
  Shanghai, China \\
  \texttt{rwang@stu.ecnu.edu.cn} \\
  \And
  Kongzhang~Hao \\
  School of computer science and software engineering \\
  UNSW, Sydney \\
  NSW, 2052 \\
  \texttt{khao@cse.unsw.edu.au} \\
  \And
  Xuemin~Lin \\
  School of computer science and software engineering \\
  UNSW, Sydney \\
  NSW, 2052 \\
  \texttt{lxue@cse.unsw.edu.au} \\
  \And
  Lu~Qin \\
  Centre for Artificial Intelligence \\
  University of Technology, Sydney \\
  NSW, 2007 \\
  \texttt{lu.qin@uts.edu.au} \\
  \And
  Wenjie~Zhang \\
  School of computer science and software engineering \\
  UNSW, Sydney \\
  NSW, 2052 \\
  \texttt{zhangw@cse.unsw.edu.au} \\
  \And
  Ying~Zhang \\
  Centre for Artificial Intelligence \\
  University of Technology, Sydney \\
  NSW, 2007 \\
  \texttt{ying.zhang@uts.edu.au} \\
  \And
  Zhengping~Qian \\
  Alibaba Group \\
  Hangzhou, China \\
  \texttt{zhengping.qzp@alibaba-inc.com} \\
    \And
  Jingren~Zhou \\
  Alibaba Group \\
  Hangzhou, China \\
  \texttt{jingren.zhou@alibaba-inc.com}
}

\begin{document}
\maketitle

\begin{abstract}
Recently there emerge many distributed algorithms that aim at solving subgraph matching at scale. Existing algorithm-level comparisons failed to provide a systematic view to the pros and cons of each algorithm mainly due to the intertwining of strategy and optimization. In this paper, we identify four strategies and three general-purpose optimizations from representative state-of-the-art works. We implement the four strategies with the optimizations based on the common \timely dataflow system for systematic strategy-level comparison. Our implementation covers all representation algorithms. We conduct extensive experiments for both unlabelled matching and labelled matching to analyze the performance of distributed subgraph matching under various settings, which is finally summarized as a practical guide. 
\end{abstract}

\keywords{Graph Pattern Matching, Query Optimization, Distributed Processing, Experiments}

\section{Introduction}
\label{sec:intro}
Given a query graph $Q$ and a data graph $G$, subgraph matching is defined as finding all subgraph instances of $G$ that are isomorphic to $Q$. In this paper, we assume that the query graph and data graph are undirected \footnote{Our implementation can seamlessly handle directed case.} simple graphs, and may be unlabelled or labelled. In this work, we mainly focus on unlabelled case given that most distributed algorithms are developed under this setting. We also demonstrate some results of labelled case due to its practical usefulness. Subgraph matching is one of the most fundamental operations in graph analysis, and has been used in a wide spectrum of applications \cite{Bi2016, Han2013, Lai2016,  Luo2008, Luo2011}. 

As subgraph matching problem is in general computationally intractable \cite{Shamir97}, and data graph nowadays is growing beyond the capacity of one single machine, people are seeking efficient and scalable algorithms in the distributed context. Unless otherwise specified, in this paper we consider a simple \textbf{hash partition} of the graph data, that is the graph is randomly partitioned by vertices, and each vertex's neighbors will be placed in the same partition. 

By treating query vertices as attributes and the matched results as relational tables, we can express subgraph matching via natural joins. The problem is accordingly transformed into seeking optimal join plan, where the optimization goal is typically to minimize the communication cost. In this paper, we focus on an in-depth survey and comparison of representative distributed subgraph matching algorithms that follow such join scheme.

\subsection{State-of-the-arts.} 
In order to solve subgraph matching using join, existing works studied various join strategies, which can be categorized into three classes, namely ``Binary-join-based subgraph-growing algorithms" (\eaat), ``Worst-case optimal vertex-growing algorithms" (\vaat) and ``Shares of Hypercube" (\multiway). We also include \other category for algorithms that do not clearly belong to the above categories. 

\stitle{\eaat.} The strategy computes subgraph matching via a series of binary joins. It first decomposes the original query graph into a set of \textit{join units} whose matches can serve the base relation of the join. Then the base relations are joined based on a predefined \textit{join order}. The \eaat algorithms differ in the selections of join unit and join order. Typical choices of join unit are star (a tree of depth 1) in \starjoin \cite{Sun2012}, \ttwig (an edge or intersection of two edges) in \ttjoin \cite{Lai2015}, and clique (a graph whose vertices are mutually connected) in \cliquejoin \cite{Lai2016}. Most existing algorithms adopt the easier-solving left-deep join order \cite{Ioannidis1991} except \cliquejoin, which explores the optimality-guaranteed bushy join \cite{Ioannidis1991}.  
 
\stitle{\vaat.} Given $\{v_0, v_1, \cdots, v_n\}$ as the query vertices, \vaat strategy first computes all matches of $\{v_0\}$ that can present in the results, then matches of $\{v_0, v_1\}$, and so forth until constructing the results. Ngo et al. proposed the worst-case optimal join algorithm \genericjoin \cite{Ngo2018}, based on which Ammar et al. implemented \bigjoin in \timely dataflow system \cite{Murray2013} and showed its worst-case optimality \cite{Ammar2018}. In this paper, we also find out that the \eaat algorithm \cliquejoin (with ``overlapped decomposition'' \footnote{Decompose the query graph into join units that are allowed to overlap edges}) is also a variant of \code{GenericJoin}, and is hence worst-case optimal. 

\stitle{\multiway.} \multiway strategy treats the computation of the query with $n$ vertices as an $n$-dimensional \textit{hypercube}. It partitions the hypercube across $w$ workers in the cluster, and then each worker can compute its own share locally with no need of exchanging data. As a result, it typically renders much less communication cost than that of \eaat and \vaat algorithms. \multiwayjoin adopts the idea of \multiway for subgraph matching. In order to properly partition the computation without missing results, \multiwayjoin needs to duplicate each edge in multiple workers. As a result, \multiwayjoin can almost carry the whole graph in each worker for certain queries \cite{Lai2015, Ammar2018} and thus scale out poorly. 

\stitle{\other.} Shao et al. proposed \psgl \cite{Shao2014} that processes subgraph matching via breadth-first-style traversal. Staring from an initial query vertex, \psgl iteratively expands the partial results by merging the matches of certain vertex's unmatched neighbors. It has been pointed out in \cite{Lai2015} that \psgl is actually a variant of \starjoin. Very recently, Qiao et al. proposed \crystaljoin \cite{Qiao2017} that aims at resolving the ``output crisis'' by compressing the (intermediate) results. The idea is to first compute the matches of the vertex cover of the query graph, then the remaining vertices' matches can be compressed as intersection of the vertex cover's neighbors to avoid costly cartesian product.

\comment{
We summarize the algorithms, the categories and their relationships in \reffig{algorithms}.

\begin{figure}[htb]
  \centering
  \includegraphics[scale=0.6]{algorithms}
  \caption{\small{The algorithms to compare in this paper.}}
  \label{fig:algorithms}
\end{figure}
}

\stitle{Optimizations.} Apart from join strategies, existing algorithms also explored a variety of optimizations, some of which are query- or algorithm-specific, while we spotlight three general-purpose optimizations, \batching, \trindex and \compression. \batching aims to divide the whole computation into sub-tasks that can be evaluated independently in order to save resource (memory) allocation. \trindex precomputes and indices the triangles (3-cycles) of the graph to facilitate pruning. \compression attempts to maintain the (intermediate) results in a compressed form to reduce resource allocation and communication cost.

\subsection{Motivations.} In this paper, we survey seven representative algorithms to solve distributed subgraph matching: \starjoin \cite{Sun2012}, \multiwayjoin \cite{AfratiFU13}, \psgl \cite{Shao2014}, \ttjoin \cite{Lai2015}, \cliquejoin \cite{Lai2016}, \crystaljoin \cite{Qiao2017} and \bigjoin \cite{Ammar2018}. While all these algorithms embody some good merits in theory, existing \textbf{algorithm-level} comparisons failed to provide a systematic view to the pros and cons of each algorithm due to several reasons. Firstly, the prior experiments did not take into consideration the differences of languages and the cost of the systems on which each implementation is based (\reftable{algorithms}). Secondly, some implementations hardcode query-specific optimizations for each query, which makes it hard to judge whether the observed performance is from the algorithmic advancement or hardcoded optimization. Thirdly, all \eaat and \vaat algorithms (more precisely, their implementations) intertwined join strategy with some optimizations of \batching, \trindex and \compression. We show in \reftable{algorithms} how each optimization has been applied in current implementation. For example, \cliquejoin only adopted \trindex and some query-specific \compression, while \bigjoin considered \batching in general, but \trindex only for one specific query (\compression was only discussed in paper, but not implemented). People naturally wonder that \textit{``maybe it is better to adopt A strategy with B optimization''}, but unfortunately none of existing implementation covers that combination. Last but not least, there misses an important benchmarking of the \rep strategy, that is to maintain the whole graph in each partition and parallelize embarrassingly \cite{Herlihy2008}. \rep strategy requires no communication, and it should be the most efficient strategy when each machine can hold the whole graph (the case for most experimental settings nowadays).

\begin{table*}[htb]
\small
  \centering
  \caption{Summarization of the surveyed algorithms.}
   \label{tab:algorithms}
   \begin{tabular}{|c|c|c|c|c|}
   \hline
   Algorithm & Category & Worst-case Optimality & Platform & Optimizations  \\
   \hline\hline
   \starjoin \cite{Sun2012} & \eaat & No & Trinity \cite{Shao2013} & None \\
   \hline
   \multiwayjoin \cite{AfratiFU13} & \multiway & N/A & Hadoop \cite{Lai2015}, Myria \cite{Chu2015}  & N/A \\
   \hline
   \psgl \cite{Shao2014} & \other & No & Giraph \cite{giraph}  & None \\
   \hline
   \ttjoin \cite{Lai2015} & \eaat & No & Hadoop & \compression \cite{Lai2017} \\
   \hline
   \cliquejoin \cite{Lai2016} & \eaat & Yes (\refsec{discussions}) & Hadoop & \trindex, some \compression \\
   \hline
   \crystaljoin \cite{Qiao2017} &\other & N/A & Hadoop & \trindex, \compression \\
   \hline
   \bigjoin \cite{Ammar2018} & \vaat & Yes \cite{Ammar2018} & Timely Dataflow \cite{Murray2013} & \batching, specific \trindex \\
   \hline
   \end{tabular}
\end{table*}

\reftable{algorithms} summarizes the surveyed algorithms via the category of strategy, the optimality guarantee, and the status of current implementations including the based platform and how the three optimizations are adopted. 

\subsection{Our Contributions}
To address the above issues, we aims at a systematic, \textbf{strategy-level} benchmarking of distributed subgraph matching in this paper. To achieve that goal, we implement all strategies, together with the three general-purpose optimizations for subgraph matching based on the \timely dataflow system \cite{Murray2013}. Note that our implementation covers all seven representative algorithms. Here, we use \timely as the base system as it incurs less cost \cite{McSherry2015} than other popular systems like Giraph \cite{giraph}, Spark \cite{Zaharia2010} and GraphLab \cite{Low2012}, so that the system's impact can be reduced to the minimum.

We implement the benchmarking platform using our best effort based on the papers of each algorithm and email communications with the authors. Our implementation is  (1) \textbf{generic} to handle arbitrary query, and does not include any hardcoded optimizations; (2) \textbf{flexible} that can configure \batching, \trindex and \compression optimizations in any combination for \eaat and \vaat algorithms; and (3) \textbf{efficient} that are comparable to and sometimes even faster than the original hardcoded implementation. Note that the three general-purpose optimizations are mainly used to reduce communication cost, and is not useful to the \multiway and \rep strategies, while we still devote a lot of efforts into their implementations. Aware that their performance heavily depends on the local algorithm, we implement and compare the state-of-the-art local subgraph matching algorithms proposed in \cite{Kim2016}, \cite{Aberger2016} (for unlabelled matching), and \cite{Bi2016} (for labelled matching), and adopt the best-possible implementation. For \multiway, we refer to \cite{Chu2015} to implement ``Hypercube Optimization'' for better hypercube sharing.

We make the following contributions in the paper.

\stitle {(1) A benchmarking platform based on \timely dataflow system for distributed subgraph matching.} We implement four distributed subgraph matching strategies (and the general optimizations) that covers seven state-of-the-art algorithms: \starjoin \cite{Sun2012},  \multiwayjoin \cite{AfratiFU13}, \psgl \cite{Shao2014}, \ttjoin \cite{Lai2015}, \cliquejoin \cite{Lai2016}, \crystaljoin \cite{Qiao2017} and \optjoin \cite{Ammar2018}. Our implementation is generic to handle arbitrary query, including the labelled and directed query, and thus can guide practical use.

\stitle{(2) Three general-purpose optimizations - \batching, \trindex and \compression.} We investigate the literature on the optimization strategies, and spotlight the three general-purpose optimizations. We propose heuristics to incorporate the three optimizations into \eaat and \vaat strategies, with no need of query-specific adjustments from human experts. The three optimizations can be flexibly configured in any combination.

\stitle{(3) In-depth experimental studies.} In order to extensively evaluate the performance of each strategy and the effectiveness of the optimizations, we use data graphs of different sizes and densities, including sparse road network, dense ego network, and web-scale graph that is larger than each machine's configured memory. We select query graphs of various characteristics that are either from existing works or suitable for benchmarking purpose. In addition to running time, we measure the communication cost, memory usage and other metrics to help reason the performance. 

\stitle{(4) A practical guide of distributed subgraph matching.} Through empirical analysis covering the variances of join strategies, optimizations, join plans, we propose a practical guide for distributed subgraph matching. We also inspire interesting future work based on the experimental findings. 

\subsection{Organizations}. The rest of the paper is organized as follows. \refsec{backgrounds} defines the problem of subgraph matching and introduces preliminary knowledge. \refsec{algorithms} surveys the representative algorithms, and our implementation details following the categories of \eaat, \vaat, \multiway and \other. \refsec{opt} investigates the three general-purpose optimizations and devises heuristics of applying them to \eaat and \vaat algorithms. \refsec{experiments} demonstrates the experimental results and our in-depth analysis. \refsec{related_works} discusses the related works, and \refsec{conclusion} concludes the whole paper.

\section{Preliminaries}
\label{sec:backgrounds}

\subsection{Problem Definition}
\label{sec:problem_definition}

\stitle{Graph Notations.} A graph $g$ is defined as a 3-tuple, $g = (V_g, E_g, L_g)$, where $V_g$ is the vertex set and $E_g \subseteq V_g \times V_g$ is the edge set of $g$, and $L_g$ is a label function that maps each vertex $\mu \in V_g$ and/or each edge $e \in E_g$ to a label. Note that for unlabelled graph, $L_g$ simply maps all vertices and edges to $\emptyset$. For a vertex $\mu \in V_g$, denote $\mathcal{N}_g(\mu)$ as the set of neighbors, $d_g(\mu) = |\mathcal{N}_g(\mu)|$ as the degree of $\mu$, $\overline{d}_g = \frac{2|E_g|}{|V_g|}$ and $D_g = \max_{\mu\in V(g)} d_g(\mu)$ as the average and maximum degree, respectively. A \textit{subgraph} $g'$ of $g$, denoted $g' \subseteq g$, is a graph that satisfies $V_{g'}\subseteq V_g$ and $E_{g'}\subseteq E_g$. 

Given $V' \subseteq V_g$, we define induced subgraph $g(V')$ as the subgraph induced by $V'$, that is $g(V') = (V', E(V'), L_g)$, where $E(V') = \{e = (\mu, \mu')\;|\; e \in E_g , \mu \in V' \land \mu' \in V'\}$. We say $V' \subseteq V_g$ is a vertex cover of $g$, if $\forall$ $e = (\mu, \mu') \in E_g$, $\mu \in V'$ or $\mu' \in V'$. A minimum vertex cover $V^c_g$ is a vertex cover of $g$ that contains minimum number of vertices. A connected vertex cover is a vertex cover whose induced subgraph is connected, among which a minimum connected vertex cover, denoted $V^{cc}_g$, is the one with the minimum number of vertices.

\stitle{Data and Query Graph.} We denote the data graph as $G$, and let $N = |V_G|$, $M = |E_G|$. Denote a data vertex of id $i$ as $u_i$ where $1 <= i <= N$. Note that the data vertex has been reordered such that if $d_G(u) < d_G(u')$, then $id(u) < id(u')$. We denote the query graph as $Q$, and let $n = |V_Q|$, $m = |E_Q|$, and $V_Q = \{v_1, v_2, \cdots, v_n\}$.

\comment{
In this work, we partition the data graph $G$ across $w$ workers in two ways:
\begin{enumerate}
\setlength\itemsep{-0.3em}
\item Hash partition, denoted $\Phi^\gamma(G)$. Given a random function $\gamma: V_G \rightarrow [1, w]$, we place $(u, \mathcal{N}_G(u))$ in the worker $\gamma(u)$. We use this strategy by default without otherwise specified.
\item Triangle partition, denoted $\Phi^\tau(G)$. On top of $\Phi^\gamma(G)$, $\forall u \in V_G$, we further include the edge $(u', u'')$ in the worker $\gamma(u)$, where $u', u'' \in \mathcal{N}_G(u), (u', u'') \in E_G$, and $id(u) < id(u') \land id(u) < id(u'')$. Obviously, $(u, u', u'')$ form a triangle in $G$ with $u$ as the smallest vertex.
\end{enumerate}
}

\stitle{Subgraph Matching.} Given a data graph $G$ and a query graph $Q$, we define \textit{subgraph isomorphism}:

\begin{definition}
	\label{def:isomorphism} \sstitle{(Subgraph Isomorphism.)} \textit{Subgraph isomorphism} is defined as a bijective mapping $f: V(Q) \rightarrow V(G)$ such that, (1) $\forall v \in V(Q)$, $L_Q(v) = L_{G}(f(v))$; (2) $\forall (v, v') \in E(Q)$, $(f(v), f(v')) \in E(G)$, and $L_Q((v, v')) = L_{G}((f(v), f(v')))$. A subgraph isomorphism is called a \textit{Match} in this paper. With the query vertices listed as $\{v_1, v_2, \cdots, v_n\}$, we can simply represent a match $f$ as $\{u_{k_1}, u_{k_2}, \cdots, u_{k_n}\}$, where $f(v_i) = u_{k_i}$ for $1 <= i <= n$.
	\end{definition}
	
	The \textit{Subgraph Matching} problem aims at finding all matches of $Q$ in $G$. Denote $R_G(Q)$ (or $R(Q)$ when the context is clear) as the result set of $Q$ in $G$. As prior works \cite{Lai2015, Lai2016, Shao2014}, we apply \textit{symmetry breaking} for unlabelled matching to avoid duplicate enumeration caused by automorphism.  Specifically, we first assign partial order $O_Q$ to the query graph according to \cite{Grochow2007}. Here, $O_Q \subseteq V_Q \times V_Q$, and $(v_i, v_j) \in O_Q$ means $v_i < v_j$. In unlabelled matching, a match $f$ must satisfy the \textit{order constraint}: $\forall (v, v') \in O_Q$, it holds $f(v) < f(v')$. Note that we do \textbf{not} consider order constraint in labelled matching. 

\begin{example}
\label{ex:subgraph_isomorphism}	
In \reffig{subgraph_isomorphism}, we present a query graph $Q$ and a data graph $G$. For unlabelled matching, we give the partial order $O_Q$ under the query graph. There are three matches: $\{u_1, u_2, u_6, u_5\}$, $\{u_2, u_5, u_3, u_6\}$ and $\{u_4, u_3, u_6, u_5\}$. It is easy to check that these matches satisfy the order constraint. Without the order constraint, there are actually four automorphic\footnote{Automorphism is an isomorphism from one graph to itself.} matches corresponding to each above match \cite{AfratiFU13}. For labelled matching, we use different fillings to represent the labels. There are two matches accordingly - $\{u_1, u_2, u_6, u_5\}$ and $\{u_4, u_3, u_6, u_5\}$.  
\end{example}

\begin{figure}[htb]
  \centering
  \includegraphics[scale=0.8]{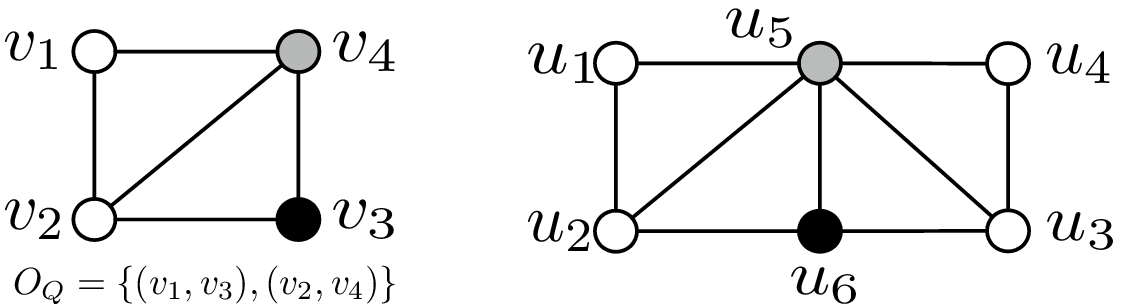}
  \caption{\small{Query Graph $Q$ (Left) and Data Graph $G$ (Right).}}
  \label{fig:subgraph_isomorphism}
\end{figure}

By treating the query vertices as attributes and data edges as relational table, we can write subgraph matching query as a multiway-way join of the edge relations. For example, regardless of label and order constraints, the query of \refex{subgraph_isomorphism} can be written as the following join

\begin{equation}
\label{eq:subgraph_isomorphism}
\small
\begin{aligned}
	R(Q) &= E(v_1, v_2) \Join E(v_2, v_3) \\
	&\Join E(v_3, v_4) \Join E(v_1, v_4) \Join E(v_2, v_4).
\end{aligned}
\end{equation}
 
This motivates researchers to leverage join operation for large-scale subgraph matching, given that join can be easily distributed, and it is natively supported in many distributed data engines like Spark \cite{Zaharia2010} and Flink \cite{Carbone2015}.

\subsection{Timely Dataflow System}
\label{sec:timely_dataflow}
\timely is a distributed data-parallel dataflow system \cite{Murray2013}. The minimum processing unit of \timely is a \textit{worker}, which can be simply seen as a process that occupies a CPU core. Typically, one physical multi-core machine can run several workers. \timely follows the \textit{shared-nothing dataflow} computation model \cite{DeWitt1990} that abstracts the computation as a dataflow graph. In the dataflow graph, the vertex (a.k.a. \textit{operator}) defines the computing logics and the edges in between the operators represent the data streams. One operator can accept multiple input streams, feed them to the computing, and produce (typically) one output stream.  
After the dataflow graph for certain computing task is defined, it is distributed to each worker in the cluster, and further translated into a physical execution plan. Based on the physical plan, each worker can accordingly process the task in parallel while accepting the corresponding input portion.

\section{Algorithm Survey}
\label{sec:algorithms}
We survey the distributed subgraph matching algorithms following the categories of \eaat, \vaat, \multiway, and \other. We also show that \cliquejoin is a variant of \genericjoin \cite{Ngo2018}, and is thus worst-case optimal.

\subsection{BinJoin}
\label{sec:edge-at-a-time}

The simplest \eaat algorithm uses data edges as the base relation, which starts from one edge, and expands by one edge in each join. For example, to solve the join of \refeq{subgraph_isomorphism}, a simple plan is shown in \reffig{simple_join_a}. The join plan is straightforward, but the intermediate results, especially $R_2$ (a 3-path), can be huge. 
\begin{figure}[htb]
    \centering
    \begin{subfigure}[b]{0.5\columnwidth}
        \centering
        \includegraphics[width=0.8\textwidth]{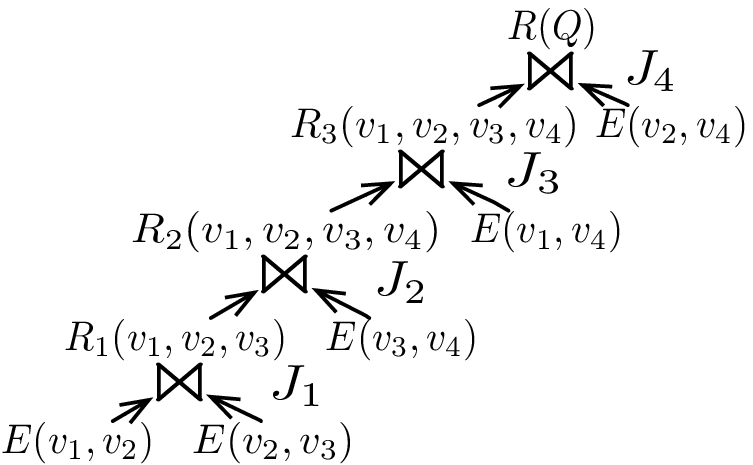}
        \caption{Left-deep join plan}
        \label{fig:simple_join_a}
    \end{subfigure}%
    \begin{subfigure}[b]{0.5\columnwidth}
        \centering
        \includegraphics[width=0.8\textwidth]{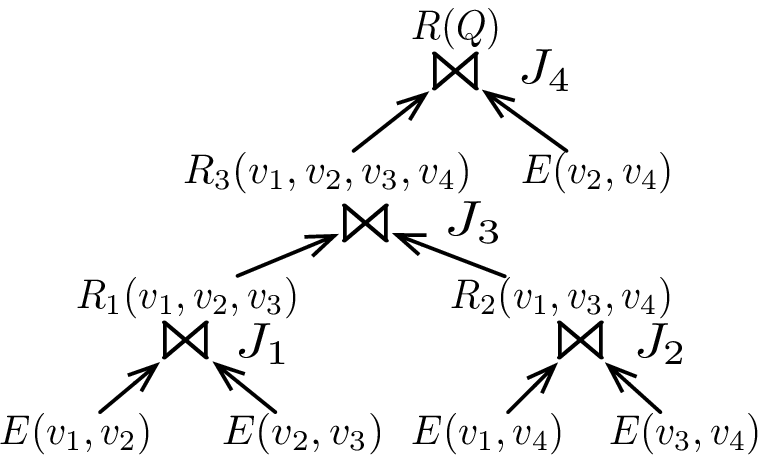}
        \caption{Bushy join plan}
        \label{fig:simple_join_b}
    \end{subfigure}
    \caption{Simple \eaat Join Plans.}
    \label{fig:simple_join}
\end{figure}

To improve the performance of \eaat, people devoted their efforts into: (1) using more complex base relations other than edge; (2) devising better join plan $P$. The base relations $B_{[q]}$ represent the matches of a set of sub-structures $[q]$ of the query graph $Q$. Each $p \in [q]$ is called a join unit, and it must satisfy $V_Q = \bigcup_{p\in[q]} V_p$ and $E_Q = \bigcup_{p\in[q]} E_p$. With the data graph partitioned across the cluster, \cite{Lai2016} constrains the join unit to be the structure whose results can be independently computed within each partition (i.e. embarrassingly parallel \cite{Herlihy2008}). It is not hard to see that when each vertex has full access to the neighbors in the partition, we can compute the matches of a $k$-star (a star of $k$ leaves) rooted on the vertex $u$ by enumerating all $k$-combinations within $\mathcal{N}_G(u)$. Therefore, star is a qualified and indeed widely used join unit.

Given the base relations, the join plan $P$ determines an order of processing binary joins. A join plan is \textit{left-deep}\footnote{More precisely it is deep, and can further be left-deep and right-deep. In this paper, we assume that it is left-deep following the prior work \cite{Lai2015}.} if there is at least a base relation involved in each join, otherwise it is \textit{bushy}. For example, the join plan in \reffig{simple_join_a} is left-deep, and a bushy join plan is shown in \reffig{simple_join_b}. Note that the bushy plan avoids the expensive $R_2$ in the left-deep plan, and is generally better. 

\stitle{StarJoin.} As the name suggests, \starjoin uses star as the join unit, and it follows the left-deep join order. 
To decompose the query graph, it first locates the vertex cover of the query graph, and each vertex in the cover and its unused neighbors naturally form a star \cite{Sun2012}. A \starjoin plan for \refeq{subgraph_isomorphism} is
\begin{equation*}
\small
(J_1)\;\;R(Q) = \Star(v_2; \{v_1, v_3, v_4\}) \Join \Star(v_4; \{v_2, v_3\}),
\end{equation*}
where $\Star(r; L)$ denotes a \Star relation (the matches of the star) with $r$ as the root, and $L$ as the set of leaves.

\stitle{TwinTwigJoin.} Enumerating a $k$-star on a vertex of degree $d$ will render $O(d^k)$ cost. We refer \textit{star explosion} to the case while enumerating stars on a large-degree vertex. Lai et al. proposed \ttjoin \cite{Lai2015} to address the issue of \starjoin by forcing the join plan to use \ttwig (a star of at most two edges) instead of a general star as the join unit. Intuitively, this would help ameliorate the star explosion by constraining the cost of each join unit from $d^k$ of arbitrary $k$ to at most $d^2$. \ttjoin follows \starjoin to use left-deep join order. 
The authors proved that \ttjoin is instance optimal to \starjoin, that is given any general \starjoin plan in the left-deep join order, we can rewrite it as an alternative \ttjoin plan that draws no more cost (in the big $O$ sense) than the original \starjoin, where the cost is evaluated based on Erd\"os-R\'enyi random graph (\er) model \cite{Erdos60onthe}.
A \ttjoin plan for \refeq{subgraph_isomorphism} is
\begin{equation}
\label{eq:ttjoin}
\small
\begin{aligned}
(J_1)\;\; & R_1(v_1, v_2, v_3, v_4) = \\
& \ttwig(v_1; \{v_2, v_4\}) \Join \ttwig(v_2; \{v_3, v_4\}); \\
(J_2)\;\; & R(Q) = R_1(v_1, v_2, v_3, v_4) \Join \ttwig(v3; \{v4\}),
\end{aligned}
\end{equation}
where $\ttwig(r; L)$ denotes a \ttwig relation with $r$ as the root, and $L$ as the leaves.

\stitle{CliqueJoin.} \ttjoin hampers star explosion to some extent, but still suffers from the problems of long execution ($\Omega(\frac{m}{2})$ rounds) and suboptimal left-deep join plan. \cliquejoin resolves the issues by extending \starjoin in two aspects. Firstly, \cliquejoin applies the ``triangle partition'' strategy (\refsec{trindex}), which enables \cliquejoin to use clique, in addition to star, as the join unit. The use of clique can greatly shorten the execution especially when the query is dense, although it still degenerates to \starjoin when the query contains no clique subgraph. Secondly, \cliquejoin exploits the bushy join plan to approach optimality. 
A \cliquejoin plan for \refeq{subgraph_isomorphism} is:
\begin{equation}
\label{eq:clique_join}
(J_1)\;\;R(Q) = \clique(\{v_1, v_2, v_4\}) \Join \clique(\{v_2, v_3, v_4\}),	
\end{equation}
where $\clique(V)$ denotes a \clique relation of the involving vertices $V$. 

\stitle{Implementation Details.} We implement the \eaat strategy based on the join framework proposed in \cite{Lai2016} to cover \starjoin, \ttjoin and \cliquejoin.

We use power-law random graph (\power) model \cite{Chung03} to estimate the cost as \cite{Lai2016}, and implement the dynamic programming algorithm \cite{Lai2016} to compute the cost-optimal join plan. 
Once the join plan is computed, we translate the plan into \timely dataflow that processes each binary join using a \code{Join} operator. We implement the \code{Join} operator following \timely's official ``pipeline'' \code{HashJoin} example\footnote{\url{https://github.com/TimelyDataflow/timely-dataflow/blob/master/examples/hashjoin.rs}}. We modify it into ``batching-style'' - the mappers (senders) shuffle the data based on the join key, while the reducers (receivers) maintain the received key-value pairs in a hash table (until mapper completes) for join processing. The reasons that we implement the join as ``batching-style'' are, (1) its performance is similar to ``pipeline'' join as a whole; (2) it replays the original implementation in Hadoop; and (3) it favors the \batching optimization (\refsec{batching}).

\subsection{WOptJoin}
\label{sec:vertex-at-a-time}
\vaat strategy processes subgraph matching by matching vertices in a predefined order. Given the query graph $Q$ and $V_Q = \{v_1, v_2, \cdots, v_n\}$ as the matching order, the algorithm starts from an empty set, and computes the matches of the subset $\{v_1, \cdots, v_i\}$ in the $i^{th}$ rounds. Denote the partial results after the $i^{th}$ ($i < n$) round as $R_i$, and $p = \{u_{k_1}, u_{k_2}, \cdots, u_{k_i}\} \in R_i$ is one of the tuples. In the $i + 1^{th}$ round, the algorithm expands the results by matching $v_{i + 1}$ with $u_{k_{i + 1}}$ for $p$ iff. $\forall_{1 \leq j \leq i} (v_j, v_{i + 1}) \in E_Q$, $(u_{k_j}, u_{k_{i + 1}}) \in E_G$. It is immediate that the candidate matches of $v_{i+1}$, denoted $C(v_{i+1})$, can be obtained by intersecting the relevant neighbors of the matched vertices as
\begin{equation}
\label{eq:intersection}
C(v_{i + 1}) = \bigcap_{\forall_{1 \leq j \leq i} \land (v_j, v_{i + 1}) \in E_Q} \mathcal{N}_G(u_{k_j}). 
\end{equation}
 
\stitle{BiGJoin.} \optjoin adopts the \vaat strategy in \timely dataflow system. The main challenge is to implement the intersection efficiently using \timely dataflow. For that purpose, the authors designed the following three operators: 
\begin{itemize}
    \setlength\itemsep{-0.3em}
	\item \code{Count}: Checking the number of neighbors of each $u_{k_j}$ in \refeq{intersection} and recording the location (worker) of the one with the smallest neighbor set.
	\item \code{Propose}: Attaching the smallest neighbor set to $p$ as $(p; C(v_{i + 1}))$.
	\item \code{Intersect}: Sending $(p; C(v_{i+1}))$ to the worker that maintains each $u_{k_j}$ and update $C(v_{i+1}) = C(v_{i+1}) \cap \mathcal{N}_G(u_{k_j})$.
\end{itemize}
After intersection, we will expand $p$ by pushing into $p$ every vertex of $C(v_{i+1})$. 

\comment{
Note that the partial results $R_i$ can be huge already, yet in the next round, \optjoin still needs to \code{Propose} on each $p \in R_i$ to get the candidate sets for the next query vertex. The memory burden piles up as the progress of the algorithm. To resolve this issue, \optjoin applies the following batching strategy: it first divides a certain $R_i$ evenly into $b$ parts, and then evaluates each part using a sub-dataflow. These sub-dataflows are physically independent, and can be scheduled to run in a batching manner. Technically speaking, we can apply such batching to any $R_i$, and even fork sub-batches. For simplicity, the authors implemented \optjoin that only batches upon $R_2$, i.e. the edges of the graph. 
}

\stitle{Implementation Details.} We directly use the authors' implementation \cite{timely-dataflow}, but slightly modify the codes to use the common graph data structure. We do not consider the dynamic version of \optjoin in this paper, as the other strategies currently only support static context. The matching order is determined using a greedy heuristic that starts with the vertex of the largest degree, and consequently selects the next vertex with the most connections (id as tie breaker) with already-selected vertices. 

\subsection{ShrCube}
\label{sec:multiway_join}
\multiway strategy treats the join processing of the query $Q$ as a hypercube of $n = |V_Q|$ dimension. It attempts to divide the hypercube evenly across the workers in the cluster, so that each worker can complete its own share without data communication. However, it is normally required that each data tuple is duplicated into multiple workers. This renders a space requirement of $\frac{M}{w^{1-\rho}}$ for each worker, where $M$ is size of the input data, $w$ is the number of workers and $0 < \rho \leq 1$ is a query-dependent parameter. When $\rho$ is close to $1$, the algorithm ends up with maintaining the whole input data in each worker. 

\stitle{MultiwayJoin.} \multiwayjoin applies the \multiway strategy to solve subgraph matching in one single round. Consider $w$ workers in the cluster, a query graph $Q$ with $V_Q = \{v_1, v_2, \ldots, v_n\}$ vertices and $E_Q = \{e_1, e_2, \ldots, e_m\}$, where $e_i = (v_{i_1}, v_{i_2})$. Regarding each query vertex $v_i$, assign a positive integer as bucket number $b_i$ that satisfies $\prod_{i=1}^n b_i = w$. The algorithm then divides the candidate data vertices for $v_i$ evenly into $b_i$ parts via a hash function $h: u \mapsto z_i$, where $u\in V_G, 1\leq z_i\leq b_i$. This accordingly divides the whole computation into $w$ shares, each of which can be indiced via an $n$-ary tuple $(z_1, z_2, \cdots, z_n)$, and is assigned to one worker. Afterwards, regarding each query edge $e_i = (v_{i_1}, v_{i_2})$, \multiwayjoin maps a data edge $(u, u')$ as $(z_1, \cdots, z_{i_1} = h(u), \cdots, z_{i_2} = h(u'), \ldots, z_n)$, where other than $z_{i_1}$ and $z_{i_2}$, each above $z_i$ iterates through $\{1,2,\cdots,b_i\}$, and the edge will be routed to the workers accordingly. Taking triangle query with $E_Q = \{(v_1, v_2), (v_1, v_3), (v_2, v_3)\}$ as an example. According to \cite{AfratiFU13}, $b_1 = b_2 = b_3 = b = \sqrt[n]{w}$ is an optimal bucket number assignment. Each edge $(u, u')$ is then routed to the workers as: (1) $(h(u), h(u'), z)$ regarding $(v_1, v_2)$; (2) $(h(u), z, h(u'))$ regarding $(v_1, v_3)$; (3) $(z, h(u), h(u'))$ regarding $(v_2, v_3)$, where the above $z$ iterates through $\{1,2,\cdots,b\}$. Consequently, each data edge is duplicated by roughly $3\sqrt[3]{w}$ times, and by expectation each worker will receive $\frac{3M}{w^{1 - 1/3}}$ edges. For unlabelled matching, \multiwayjoin utilizes the partial order of the query graph (\refsec{problem_definition}) to reduce edge duplication, and details can be found in \cite{AfratiFU13}. 

\stitle{Implementation Details.} There are two main impact factors of the performance of \multiway. Firstly, the hypercube sharing by assigning proper $b_i$ for $v_i$. Beame et al. \cite{Beame2017} generalized the problem of computing optimal hypercube sharing for arbitrary query as linear programming. However, the optimal solution may assign fractional bucket number that is unwanted in practice. An easy refinement is to round down to an integer, but it will apparently result in idle workers. Chu et al. \cite{Chu2015} addressed this issue via ``Hypercube Optimization'', that is to enumerate all possible bucket sequences around the optimal solutions, and choose the one that produces shares (product of bucket numbers) closest to the number of workers. We adopt this strategy in our implementation.

Secondly, the local algorithm. When the edges arrive at the worker, we collect them into a local graph (duplicate edges are removed), and use local algorithm to compute the matches. For unlabelled matching. we study the state-of-the-art local algorithms from ``EmptyHeaded'' \cite{Aberger2016} and ``DualSim''\cite{Kim2016}. ``EmptyHeaded'' is inspired by Ngo's worst-case optimal algorithm \cite{Ngo2018} that decomposes the query graph via ``Hyper-Tree Decomposition'', computes each decomposed part using worst-case optimal join and finally glues all parts together using hash join. ``DualSim'' was proposed by \cite{Kim2016} for subgraph matching in the external-memory setting. The idea is to first compute the matches of $V_Q^{cc}$, then the remaining vertices $V_Q \setminus V_Q^{cc}$ can be efficiently matched by enumerating the intersection of $V_Q^{cc}$'s neighbors. We find out that ``DualSim'' actually produces the same query plans as ``EmptyHeaded'' for all our benchmarking queries (\reffig{queries}) except $q_9$. We implement both algorithms for $q_9$ and ``DualSim'' performs better than ``EmptyHeaded'' on the GO, US, GP and LJ datasets (\reftable{datasets}). 
As a result, we adopt ``DualSim'' as the local algorithm for \multiwayjoin. For labelled matching, we implement ``CFLMatch'' proposed in \cite{Bi2016} that has been shown so far to have the best performance.  

Now we let each worker independently compute matches in its local graph. Simply doing so will result in duplicates, so we process deduplication as follows: given a match $f$ that is computed in the worker identified by $t_w$, we can recover the tuple $t^f_{e}$ of the matched edge $(f(v), f(v'))$ regarding the query edge $e = (v, v')$, then the match $f$ is retained if and only if $t_w = t^f_{e}$ for every $e \in E_Q$. To explain this, let's consider $b = 2$, and a match $\{u_0, u_1, u_2\}$ for a triangle query $(v_0, v_1, v_2)$, where $h(u_0) = h(u_1) = h(u_2) = 0$. It is easy to see that the match will be computed in workers of $(0, 0, 0)$ and $(0, 0, 1)$, while the match in worker $(0, 0, 1)$ will be eliminated as $(u_0, u_2)$ that matches the query edge $(v_0, v_2)$ can not be hashed to $(0, 0, 1)$ regarding $(v_0, v_2)$. We can also avoid deduplication by separately maintaing each edge regarding different query edges it stands for, and use the local algorithm proposed in \cite{Chu2015}, but it results in too many edge duplicates that drain our memory even when processing a medium-size graph.

 \subsection{Others}
 \label{sec:others}
 \stitle{PSgL and its implementation.} \psgl iteratively processes subgraph matching via breadth-first traversal. All query vertices are configured three status, ``white'' (initialized), ``gray'' (candidate) and ``black'' (matched). Denote $v_i$ as the vertex to match in the $i^{th}$ round. The algorithm starts from matching initial query vertex $v_1$, and coloring the neighbors as ``gray''. In the $i^{th}$ round, the algorithm applies the workload-aware expanding strategy at runtime, that is to select the $v_i$ to expand among all current ``gray'' vertices based on a greedy heuristic to minimize the communication cost \cite{Shao2013}; the partial results from previous round $R_{i-1}$ (specially, $R_0 = \emptyset$) will be distributed among the workers based on the candidate data vertices that can match $v_i$; in the certain worker, the algorithm computes $R_i$ by merging $R_{i-1}$ with the matches of the \Star formed by $v_i$ and its ``white'' neighbors $\mathcal{N}^w_Q(v_i)$, namely $\Star(v_i; \mathcal{N}^w_Q(v_i))$; after $v_i$ is matched, $v_i$ is colored as ``black'' and its ``white'' neighbors will be colored as ``gray''; essentially, this process is analogous to \starjoin by processing $R_i = R_{i-1} \Join \Star(v_i; \mathcal{N}^w_Q(v_i))$. Thus, \psgl can be seen as an alternative implementation of \starjoin on \code{Pregel} \cite{Malewicz10}. In this work, we also implement \psgl using a \code{Pregel} on \timely. Note that we introduce \code{Pregel} api to as much as possible replay the implementation of \psgl. In fact, it is simply wrapping \timely's primitive operators such as \code{binary\_notify} and \code{loop} \footnote{\url{https://github.com/frankmcsherry/blog/blob/master/posts/2015-09-21.md}}, and barely introduces extra cost to the implementation. Our experimental results demonstrate similar findings as prior work \cite{Lai2016} that \psgl's performance is dominated by \cliquejoin \cite{Lai2016}. Thus, we will not further discuss this algorithm in this paper.


\stitle{CrystalJoin and its implementation.} \crystaljoin aims at resolving the ``output crisis'' by compressing the results of subgraph matching \cite{Qiao2017}. The authors defined a structure called \textit{crystal}, denoted $\mathcal{Q}(x, y)$. A crystal is a subgraph of $Q$ that contains two sets of vertices $V_x$ and $V_y$ ($|V_x| = x$ and $|V_y| = y$), where the induced subgraph $Q(V_x)$ is a $x$-clique, and every vertex in $V_y$ connects to all vertices of $V_x$. We call $V_x$ clique vertices, and $V_y$ the bud vertices. The algorithm first obtains the minimum vertex cover $V_Q^c$, and then applies the \textit{Core-Crystal Decomposition} to decompose the query graph into the \textit{core} $Q(V^c_Q)$ and a set of \textit{crystal}s $\{\mathcal{Q}_1(x_1, y_1), \ldots, \mathcal{Q}_t(x_t, y_t)\}$. The crystals must satisfy that $\forall 1 \leq i \leq t$, $Q(V_{x_i}) \subseteq Q(V_Q^c)$, namely, the clique part of each crystal is a subgraph of the core. As an example, we plot a query graph and the corresponding core-crystal decomposition  in \reffig{core_crystals}. Note that in the example, both crystals have an edge (i.e. 2-clique) as the clique part. 
\begin{figure}[htb]
  \centering
  \includegraphics[scale=0.95]{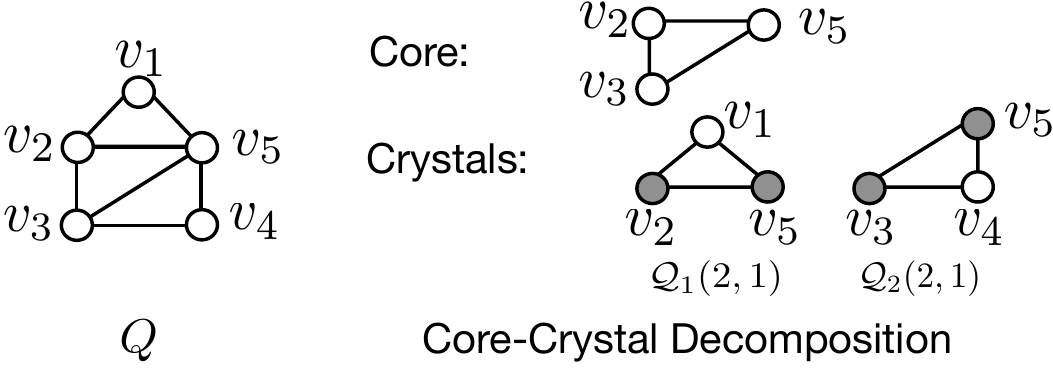}
  \caption{\small{The Core-Crystal Decomposition of the query graph.}}
  \label{fig:core_crystals}
\end{figure}

With core-crystal decomposition, the computation has accordingly split into three stages:
\begin{enumerate}
    \setlength\itemsep{-0.3em}
	\item \textbf{Core computation}. Given that $Q(V_Q^c)$ itself is a query graph, the algorithm can be recursively applied to compute $Q(V_Q^c)$ according to \cite{Qiao2017}.
	\item \textbf{Crystal computation}. A special case of crystal is $\mathcal{Q}(x, 1)$, which is indeed a $(x+1)$-clique. Suppose an instance of the $Q(V_x)$ is $f_x = \{u_1, u_2, \ldots, u_x\}$, we can represent the matches w.r.t. $f_x$ as $(f_x, I_y)$, where $I_y = \bigcap_{i=1}^x \mathcal{N}_G(u_i)$ denotes the set of vertices that can match $V_y$. This can naturally be extended to the case with $y > 1$, where any $y$-combinations of the vertices of $I_y$ together with $f_x$ represent a match. This way, the matches of crystals can be largely compressed. 
	\item \textbf{One-time assembly}. This stage assembles the core instances and the compressed crystal matches to produce the final results. More precisely, this stage is to \textbf{join} the core instance with the crystal matches.
\end{enumerate}

We notice two technical obstacles to implement \crystaljoin according to the paper. Firstly, it is worth noting that the core $Q(V_Q^c)$ may be disconnected, a case that can produce exponential number of results. The authors applied a query-specific optimization in the original implementation to resolve this issue. Secondly, the authors proposed to precompute the cliques up to certain $k$, while it is often cost-prohibitive to do so in practice. Take UK (\reftable{datasets}) dataset as an example, the triangles, 4-cliques and 5-cliques are respectively about $20$, $600$ and $40000$ times larger than the graph itself. It is worth noting that the main purpose of this paper is not to study how well each algorithm performs for a specific query, which has its theoretical value, but can barely guide practice. After communicating with the authors, we adapt \crystaljoin in the following. Firstly, we replace the core $Q(V_Q^c)$ with the induced subgraph of the minimum connected vertex cover $Q(V_Q^{cc})$. Secondly, instead of implementing \crystaljoin as a strategy, we use it as an alternative join plan (matching order) for \vaat. According to \crystaljoin, we first match $V_Q^{cc}$, while the matching order inside and outside $V_Q^{cc}$ still follows \vaat's greedy heuristic (\refsec{vertex-at-a-time}). It is worth noting that this adaptation achieves high performance comparable to the original implementation. In fact, we also apply \crystaljoin plan to \eaat, while it does not perform as well as the \vaat version, thus we do not discuss this implementation.

\stitle{FullRep and its implementation.} \rep simply maintains a full replica of the graph in each physical machine. Each worker picks one independent share of computation and solves it using existing local algorithm.

The implementation is straightforward. We let each worker pick its share of computation via a Round-Robin strategy, that is we settle an initial query vertex $v_1$, and let first worker match $v_1$ with $u_1$ to continue the remaining process, and second worker match $v_1$ with $u_2$, and so on. This simple strategy already works very well on balancing the load of our benchmarking queries (\reffig{queries}). We use ``DualSim'' for unlabelled matching and ``CFLMatch'' for labelled matching as \multiwayjoin.

\subsection{Worst-case Optimality.}
\label{sec:worst_case_opt}
Given a query $Q$ and the data graph $G$, we denote the maximum possible result set as $\overline{R}_G(Q)$. Simply speaking, an algorithm is worst-case optimal if the aggregation of the total intermediate results is bounded by $\Theta(|\overline{R}_G(Q)|)$. Ngo et al. proposed a class of worst-case optimal join algorithm called \genericjoin \cite{Ngo2018}, and we first overview this algorithm. 

\stitle{GenericJoin.} Let the join be $R(V) = \Join_{F \subseteq \Psi} R(F)$, where $\Psi = \{U\;|\;U \subseteq V\}$ and $V = \bigcup_{U\in \Psi} U$. Given a vertex subset $U \subseteq V$, let $\Psi_U = \{V'\;|\;V' \in \Psi \land V' \cap U \neq \emptyset\}$, and for a tuple $t \in R(V)$, denote $t_U$ as $t$'s projection on $U$. We then show the \code{GenericJoin} in \refalg{generic_join}. 

\begin{algorithm}[htb]
\SetAlgoVlined
\SetFuncSty{textsf}
\SetArgSty{textsf}
\small
\caption{\code{GenericJoin}$(V, \Psi, \Join_{U \in \Psi}R(U))$}
\label{alg:generic_join}
\State{$R(V) \leftarrow \emptyset$;} \\
\If{$|V| = 1$} {
\State{\textbf{Return} $\bigcup_{U \in \Psi} R(U)$;}
}
\State{$V \leftarrow (I, J)$, where $\emptyset \neq I \subset V$, and $J = V \setminus I$;} \\
\State{$R(I) \leftarrow \code{GenericJoin}(I, \Psi_I, \Join_{U \in \Psi_I} \pi_I(R(U)))$;} \\
\ForAll{$t_I \in R(I)$} {
\State{$R(J)_{w.r.t.\;t_I} \leftarrow \code{GenericJoin}(J, \Psi_J, \Join_{U \in \Psi_J} \pi_J(R(U) \ltimes t_I))$; } \\
\State{$R(V) \leftarrow R(V) \cup \{t_I\} \times R(J)_{w.r.t.\;t_I}$;}
}
\State{\textbf{Return} $R(V)$;}
\end{algorithm}

In \refalg{generic_join}, the original join is recursively decomposed into two parts $R(I)$ and $R(J)$ regarding the disjoint sets $I$ and $J$.  From line~5, it is clear that $R(I)$ will record $R(V)$'s projection on $I$, thus we have $|R(I)| \leq |\overline{R}(V)|$, where $\overline{R}(V)$ is the maximum possible results of the query. Meanwhile, in line~7, the semi-join $R(U) \ltimes t_I = \{r\;|\; r \in R(U) \land r_{(U \cup I)} = t_{(U \cup I)}\}$ only retains those $R(J)$ w.r.t. $t_I$ that can end up in the join result, which infers that the  $R(J)$ must also be bounded by the final results. This intuitively explains the worst-case optimality of \code{GenericJoin}, while we refer interested readers to \cite{Ngo2018} for a complete proof. 

It is easy to see that \bigjoin is worst-case optimal. In \refalg{generic_join}, we select $I$ in line~4 by popping the edge relation $E(v_s, v_i) (s < i)$ in the $i^{th}$ step. In line~7, the recursive call to solve the semi-join $R(U) \ltimes t_I$ actually corresponds to the intersection process. 

\stitle{Worst-case Optimality of CliqueJoin.} Note that the two clique relations in \refeq{clique_join} interleave one common edge $(v_2, v_4)$ in the query graph. This optimization, called ``overlapping decomposition'' \cite{Lai2016}, eventually contributes to \cliquejoin's worst-cast optimality. Note that it is not possible to apply this optimization to \starjoin and \ttjoin. We have the following theorem.

\begin{theorem}
\small
\label{thm:clique_join_optimal}
\cliquejoin is worst-case optimal while applying ``overlapped decomposition''.
\end{theorem}

\begin{proof}
 We implement \cliquejoin using \refalg{generic_join} in the following. Note that $Q(V)$ denotes a subgraph of $Q$ induced by $V$. In line~2, we change the stopping condition to ``$Q(I)$ is either a clique or a star''. In line~4, the $I$ is selected such that $Q(I)$ is either a clique or a star. Note that by applying the ``overlapping decomposition'' in \cliquejoin, the sub-query of the $J$ part must be the $J$-induced graph $Q(J)$, and it will also include the edges of $E_{Q(I)} \cap E_{Q(J)}$, which infers that $R(Q(J)) = R(Q(J)) \ltimes R(Q(I))$ , and just reflects the semi-join in line~7. Therefore, \cliquejoin belongs to \code{GenericJoin}, and is thus worst-case optimal.
\end{proof}

\section{Optimizations}
\label{sec:opt}
We introduce the three general-purpose optimizations, \batching, \trindex and \compression in this section, and how we orthogonally apply them to \eaat and \vaat algorithms. In the rest of the paper, we will use the strategy \eaat, \vaat, \multiway instead of their corresponding algorithms, as we focus on strategy-level comparison. 

\subsection{Batching}
\label{sec:batching}
Let $R(V_i)$ be the partial results that match the given vertices $V_i = \{v_{s_i}, v_{s_2}, \ldots, v_{s_i}\}$ ($R_i$ for short if $V_i$ follows a given order), and $R(V_j)$ denote the more complete results with $V_i \subset V_j$. Denote $R_j|R_i$ as the tuples in $R_j$ whose projection on $V_i$ equates $R_i$. Let's partition $R_i$ into $b$ \textbf{disjoint} parts $\{R^1_i, R^2_i, \ldots, R^b_i\}$. We define \batching on $R_j|R_i$ as the technique to independently process the following sub-tasks that compute $\{R_j|R_i^1, R_j|R_i^2, \ldots, R_j|R_i^b\}$. Obviously, $R_j|R_i = \bigcup_{k = 1}^b R_j | R_i^k$. 

\stitle{WOptJoin.} Recall from \refsec{vertex-at-a-time} that \vaat progresses according to a predefined matching order $\{v_1, v_2, \ldots, v_n\}$. In the $i^{th}$ round, \vaat will \code{Propose} on each $p \in R_{i-1}$ to compute $R_i$. It is not hard to see that we can easily apply \batching to the computation of $R_i|R_{i-1}$ by randomly partitioning $R_{i-1}$. For simplicity, the authors implemented \batching on $R(Q) | R_1(v_1)$. Note that $R_1(v_1) = V_G$ in unlabelled matching, which means that we can achieve \batching simply by partitioning the data vertices\footnote{Practically, it is more efficient to start from matching the edges instead of the vertices, and we can batch on $R(Q)|R_2$, where $R_2 = E_G$.}. For short, we also say the strategy batches on $v_1$, and call $v_1$ the batching vertex. We follow the same idea to apply \batching to \eaat algorithms.


\stitle{BinJoin.} While it is natural for \vaat to batch on $v_1$, it is non-trivial to pick such a vertex for \eaat. Given a decomposition of the query graph $\{p_1, p_2, \ldots, p_s\}$, where each $p_i$ is a join unit, we have $R(Q) = R(p_1) \Join R(p_2) \cdots \Join R(p_s)$. If we partition $R_1(v)$ so as to batch on $v \in V_Q$, we correspondingly split the join task, and one of the sub-task is $R(Q) | R_1^k(v) = R(p_1) | R_1^k(v) \Join \cdots \Join R(p_s) | R_1^k(v)$ ($R_1^k(v)$ is one partition of $R_1(v)$). Observe that if there exists a join unit $p$ where $v \not \in V_p$, we must have $R(p) = R(p) | R_1^k(v)$, which means $R(p)$ have to be fully computed in each sub-task. Let's consider the example query in \refeq{ttjoin}.
\begin{equation*}
    R(Q) = T_1(v_1, v_2, v_4) \Join T_2(v_2, v_3, v_4) \Join T_3(v_3, v_4).
\end{equation*}
Suppose we batch on $v_1$, the above join can be divided into the following independent sub-tasks:
\begin{equation*}
    \small
    \begin{split}
    R(Q)|R_1^1(v_1) &= (T_1(v_1, v_2, v_4)|R_1^1(v_1)) \Join T_2(v_2, v_3, v_4) \Join T_3(v_3, v_4), \\
    R(Q)|R_1^2(v_1) &= (T_1(v_1, v_2, v_4)|R_1^2(v_1)) \Join T_2(v_2, v_3, v_4) \Join T_3(v_3, v_4), \\
    &\cdots \\
    R(Q)|R_1^b(v_1) &= (T_1(v_1, v_2, v_4)|R_1^b(v_1)) \Join T_2(v_2, v_3, v_4) \Join T_3(v_3, v_4).
    \end{split}
\end{equation*}
It is not hard to see that we will have to re-compute $T_2(v_2, v_3, v_4)$ and $T_3(v_3, v_4)$ in all the above sub-tasks. Alternatively, if we batch on $v_4$, we can avoid such re-computation as $T_1, T_2$ and $T_3$ can all be partitioned in each sub-task. Inspired by this, for \eaat, we come up with the heuristic to apply \batching on the vertex that presents in as many join units as possible. Note that such vertex can only be in the join key, as otherwise it must at least not present in one side of the join. For complex query, we can still have join unit that does not contain any vertex for \batching after applying the above heuristic. In this case, we either re-compute the join unit, or cache it on disk. Another problem caused by this is potential memory burden of the join. Thus, we devise the \textit{join-level} \batching following the idea of external \code{MergeSort}. Specifically, we inject a \code{Buffer-and-Batch} operator for the two data streams before they arrive at the \code{Join} operator. \code{Buffer-and-Batch} functions in two parts:
\begin{itemize}
    \setlength\itemsep{-0.3em}
	\item \code{Buffer}: While the operator receives data from the upstream, it buffers the data until reaching a given threshold. Then the buffer is sorted according to the join key's hash value and spilled to the disk. The buffer is reused for the next batch of data.
	\item \code{Batch}:  After the data to join is fully received, we read back the data from the disk in a batching manner, where each batch must include all join keys whose hash values are within a certain range.
\end{itemize}
While one batch of data is delivered to the \code{Join} operator, \code{Timely} allows us to supervise the progress and hold the next batch until the current batch completes. This way, the internal memory requirement is one batch of the data. Note that such join-level \batching is natively implemented in Hadoop's ``Shuffle'' stage, and we replay this process in \timely to improve the scalability of the algorithm.

\subsection{Triangle Indexing}
\label{sec:trindex}
As the name suggests, \trindex precomputes the triangles of the data graph and indices them along with the graph data to prune infeasible results. The authors of \bigjoin \cite{Ammar2018} optimized the 4-clique query by using the triangles as base relations to join, which reduces the rounds of join and network communication. In \cite{Qiao2017}, the authors proposed to not only maintain triangles, but all $k$-cliques up to a given $k$. As we mentioned earlier, it incurs huge extra cost of maintaining triangles already, let alone larger cliques. 

In addition to the default hash partition, Lai et al. proposed ``triangle partition'' \cite{Lai2016} by also incorporating the edges among the neighbors (it forms triangles with the anchor vertex) in the partition. ``Triangle partition'' allows \eaat to use clique as the join unit \cite{Lai2016}, which greatly reduces the intermediate results of certain queries and improves the performance. ``Triangle partition'' is in de facto a variant of \trindex, which instead of explicitly materializing the triangles, maintains them in the local graph structure (e.g. adjacency list). As we will show in the experiment (\refsec{experiments}), this will save a lot of space compared to explicit triangle materialization. Therefore, we adopt the ``triangle partition'' for \trindex optimization in this work.

\stitle{BinJoin.} Obviously, \eaat becomes \cliquejoin with \trindex, and \starjoin (or \ttjoin) otherwise. With worst-case optimality guarantee (\refsec{worst_case_opt}), \eaat should perform much better with \trindex, which is also observed in ``Exp-1'' of \refsec{experiments}.


\stitle{WOptJoin.} In order to match $v_i$ in the $i^{th}$ round, \vaat utilizes \code{Count}, \code{Propose} and \code{Intersect} to process the intersection of \refeq{intersection}. For ease of presentation, suppose $v_{i+1}$ connects to the first $s$ query vertices $\{v_1, v_2, \ldots, v_s\}$, and given a partial match, $\{f(v_1), \ldots, f(v_s)\}$, we have $C(v_{i + 1})=\bigcap_{j = 1}^s \mathcal{N}_G(f(v_j))$. In the original implementation, it is required to send $(p; C(v_{i + 1}))$ via network to all machines that contain each $f(v_j) (1 \leq j \leq s)$ to process the intersection, which can render massive communication cost. In order to reduce the communication cost, we implement \trindex for \vaat in the following. We first group $\{v_1, \ldots, v_s\}$ such that for each group $U(v_x)$, we have 
\begin{equation*}
U(v_x) = \{v_x\} \cup \{v_y \;|\; (v_x, v_y) \in E_Q\}. 
\end{equation*}
Because of \trindex, we have $\mathcal{N}_G(f(v_y))$ ($\forall v_y \in U(v_x)$) maintain in $f(v_x)$'s partition. Thus, we only need to send the prefix to $f(v_x)$'s machine, and the intersection within $U(v_x)$ can be done locally. We process the grouping using a greedy strategy that always constructs the largest group from the remaining vertices.

\begin{remark}
\label{rem:triangle_partition}
The ``triangle partition'' may result in maintaining a large portion of the data graph in certain partition. Lai et al. pointed out this issue, and proposed a space-efficient alternative by leveraging the vertex orderings \cite{Lai2016}. That is, given the partitioned vertex as $u$, and two neighbors $u'$ and $u''$ that close a triangle, we place the edge $(u', u'')$ in the partition only when $u < u' < u''$. Although this alteration reduces storage, it may affect the effectiveness of \trindex for \vaat and the implementations of \batching and \compression for \eaat algorithms. Take \vaat as an example, after using the space-efficient ``triangle partition'', we should modify the above grouping as:
\begin{equation*}
    U(v_x) = \{v_x\} \cup  \{v_y \;|\; (v_x, v_y) \in E_Q \land (v_x, v_y) \in O_Q \}.
\end{equation*}
Note that the order between query vertices are for symmetry breaking (\refsec{problem_definition}), and it may not present in certain query, which makes \trindex completely useless for \vaat.
\end{remark}

\subsection{Compression}
\label{sec:compression}
Subgraph matching is a typical combinatorial problem, and can easily produce results of exponential size. \compression aims to maintain the (intermediate) results in a compressed form to reduce resource allocation and communication cost. In the following, when we say ``compress a query vertex'', we mean maintaining its matched data vertices in the form of an array, instead of unfolding them in line with the one-one mapping of a match (\refdef{isomorphism}). 
Qiao et al. proposed \crystaljoin to study \compression in general for subgraph matching. As we introduced in \refsec{others}, \crystaljoin first extracts the minimum vertex cover as uncompressed part, and then it can compress the remaining query vertices as the intersection of certain uncompressed matches' neighbors. Such \compression leverages the fact that all dependencies (edges) of the compressed part that requires further computation are already covered by the uncompressed part, thus it can stay compressed until the actual matches are requested. \crystaljoin inspires a heuristic for doing \compression, that is \textit{to compress the vertices whose matches will \textbf{not} be used in any future computation}. In the following, we will apply the same heuristic to the other algorithms. 

\stitle{BinJoin.} Obviously we can not compress any vertex that presents in the join key. What we need to do is to simply locate the vertices to compress in the join unit, namely star and clique. For star, the root vertex must remain uncompressed, as the leaves' computation depends on it. For clique, we can only compress one vertex, as otherwise the mutual connection between the compressed vertices will be lost. In a word, we compress two types of vertices for \eaat, (1) non-key and non-root vertices of a star join unit, (2) one non-key vertex of a clique join unit.

\stitle{WOptJoin.} Based on a predefined join order $\{v_1, v_2, \ldots, v_n\}$, we can compress $v_i$ ($1 \leq i \leq n$), if there does not exist $v_j$ ($i < j$) such that $(v_i, v_j) \in E_Q$. In other words, $v_i$'s matches will never be involved in any future intersection (computation). Note that $v_n$'s can be trivially compressed. With \compression, when $v_i$ is compressed, we will maintain its matches as an array instead of unfolding it into the prefix like a normal vertex.


\section{Experiments}
\label{sec:experiments}
\subsection{Experimental settings}
\label{sec:experimental_settings}
\stitle{Environments.} We deploy two clusters for the experiments: (1) a local cluster of 10 machines connected via one 10GBps switch and one 1GBps switch. Each machine has 64GB memory, 1 TB disk and 1 Intel Xeon CPU E3-1220 V6 3.00GHz with 4 physical cores; (2) an AWS cluster of 40 ``r5-2xlarge'' instances connected via a 10GBps switch, each with 64GB memory, 8 vCpus and 500GB Amazon EBS storage. By default we use the local cluster of 10 machines with 10GBps switch. We run 3 workers in each machine in the local cluster, and 6 workers in the AWS cluster for \timely. The codes are implemented based on the open-sourced \timely dataflow system \cite{timely} using Rust 1.32. We are still working towards open-sourcing the codes, and the bins together with their usages are temporarily provided\footnote{\url{https://goo.gl/Xp5BrW}} to verify the results. 

\stitle{Metrics.}  In the experiments, we measure query time $T$ as the slowest worker's wall clock time from an average of three runs. We allow 3 hours as the maximum running time for each test. We use \timeout and \oom to indicate a test case runs out of the time limit and out of memory, respectively. By default we will not show the \oom results for clear presentation. We divide $T$ into two parts, the computation time $T_{comp}$ and the communication time $T_{comm}$. We measure $T_{comp}$ as the time the slowest worker spends on actual computation by timing every computing function. We are aware that the actual communication time is hard to measure as \timely overlaps computation and communication to improve throughput. We consider $T - T_{comp}$, which mainly records the time the worker waits data from the network channel (a.k.a. communication time). While the other part of communication that overlaps computation is of less interest as it does not affect the query progress. As a result, we simply let $T_{comm} = T - T_{comp}$ in the experiments. We measure the maximum peak memory using Linux's ``\code{time -v}'' in each machine. We define the communication cost as the number of integers a worker receives during the process, and measure the maximum communication cost among the workers accordingly. 

\stitle{Dataset Formats.} We preprocess each dataset as follows: we treat it as a simple undirected graph by removing self-loop and duplicate edges, and format it using ``Compressed Sparse Row'' (CSR) \cite{csrwiki}. We relabel the vertex id according to the degree and break the ties arbitrarily. 

\stitle{Compared Strategies.} In the experiments, we implement \eaat and \vaat with all \batching, \trindex and \compression optimizations (\refsec{opt}). \multiway is implemented with ``Hypercube Optimization'' \cite{Chu2015}, and ``DualSim'' (unlabelled) \cite{Kim2016} and ``CFLMatch'' (labelled) \cite{Bi2016} as local algorithms. \rep is implemented with the same local algorithms as \multiway.

\stitle{Auxiliary Experiments.} We have also conducted several auxiliary experiments in the appendix to study the strategies of \eaat, \vaat, \multiway and \rep.

\subsection{Unlabelled Experiments}
\stitle{Datasets.} The datasets used in this experiment are shown in \reftable{datasets}. All datasets except SY are downloaded from public source, which are indicated by the letter in the bracket (S \cite{snap}, W \cite{webgraph}, D \cite{challenge9}). All statistics are measured as $G$ is an undirected graph. Among the datasets, GO is a small dataset to study cases of extremely large (intermediate) result set; LJ, UK and FS are three popular datasets used in prior works, featuring statistics of real social network and web graph; GP is the google plus ego network, which is exceptionally dense; US and EU, on the other end, are sparse road networks. These datasets vary in number of vertices and edges, densities and maximum degree, as shown in \reftable{datasets}. We synthesize the SY data according to \cite{Chakrabarti2004} that generates data with real-graph characteristics. Note that the data occupies roughly 80GB space, and is larger than the configured memory of our machine. We synthesize the data because we do not find public accessible data of this size. Larger dataset like Clueweb \cite{clubweb} is available, but it is beyond the processing power of our current cluster. 

Each data is hash partitioned (``hash'') across the cluster. We also implement the ``triangle partition'' (``tri.'') for \trindex optimization (\refsec{trindex}). To do so, we use \bigjoin to compute the triangles and send the triangle edges to corresponding partition. We record the time $T_*$ and average number of edges $|\overline{E_*}|$ of the two partition strategies. The partition statistics are recorded using the local cluster, except for SY that is processed in the AWS cluster. From \reftable{datasets}, 
we can see that $|\overline{E_{tri.}}|$ is noticeably larger, around 1-10 times larger than $|\overline{E_\text{hash}}|$. Note that in GP and UK, which either is dense, or must contain a large dense community, the ``triangle partition'' can maintain a large portion of data in each partition. While compared to complete triangle materialization, ``triangle partition'' turns out to be much cheaper. For example, the UK dataset contains around 27B triangles, which means each partition in our local cluster should by average take 0.9B triangles (three integers); in comparison, UK's ``triangle partition'' only maintains an average of 0.16B edges (two integers) according to \reftable{datasets}. 

We use US, GO and LJ as default datasets in the experiments ``Exp-1'', ``Exp-2'' and ``Exp-3'' in order to collect useful feedbacks from successful queries, while we may not present certain cases when they do not give new findings. 

\comment{
\begin{table*}
    \small
    \centering
     \begin{tabular}{|c|c|c|c|c|c|c|c|c|c|c|c|} 
     \hline
     Datasets & Name & $|V_G|$/mil & $|E_G|$/mil & $\overline{d}_G$ & $D_G$ &$T_\text{hash}$/s & $|\overline{E_\text{hash}}|$/mil &$T_\text{spa.}$/s & $|\overline{E_\text{spa.}}|$/mil &$T_\text{tri.}$/s & $|\overline{E_\text{tri.}}|$ /mil \Ts\Bs \\
     \hline\hline
     google(S) & GO & 0.86 & 4.32 & 10.04 & 6,332 & 7.5 & 0.28 & 2.07 & 0.56 & 2.98 & 2.98 \\
    \hline
    gplus(S) & GP & 0.11 & 12.23 & 445.1 & 20,127  & 5.69 & 0.80 & 50.5 & 5.24 & 101.8 & 10.68 \\
    \hline
    usa-road(D) & US & 23.95 & 28.85 & 2.41 & 9  & 23.50 & 1.89 & 3.14 & 1.89 & 4.22 & 1.90 \\
    \hline
    livejournal(S) & LJ & 4.85 & 43.37 & 17.88 & 20,333  & 16.80 & 2.81 & 24.14 & 6.73 & 49.87 & 12.49\\
    \hline
    uk2002(W) & UK & 18.50 & 298.11 & 32.23 & 194,955  & 70.10 & 17.16 & 362.70 & 46.72 & 957.24 & 156.05 \\
    \hline
    eu-road(D) & EU & 173.80 & 342.70 & 3.94 & 20 & 532.61 & 22.47 & 43.57 & 22.47& 47.73 & 22.98 \\
    \hline
    friendster(S) & FS & 65.61 & 1806.07 & 55.05 & 5,214  & 325.05 & 118.40 & 835.14 & 217.49 & 1506.95 & 395.31 \\
    \hline
    Synthetic & SY & 372.00 & 10,000.00 & 53 & 613,461  & 2027 & 493.75 & - & 574.12 & 5604.00 & 660.61 \\
    \hline
    \end{tabular}
    \caption{The unlabelled datasets.}
    \label{tab:datasets}
\end{table*}
}

\begin{table*}
    \small
    \centering
     \caption{The unlabelled datasets.}
    \label{tab:datasets}
     \begin{tabular}{|c|c|c|c|c|c|c|c|c|c|} 
     \hline
     Datasets & Name & $|V_G|$/mil & $|E_G|$/mil & $\overline{d}_G$ & $D_G$ &$T_\text{hash}$/s & $|\overline{E_\text{hash}}|$/mil &$T_\text{tri.}$/s & $|\overline{E_\text{tri.}}|$ /mil  \\
     \hline\hline
     google(S) & GO & 0.86 & 4.32 & 5.02 & 6,332 & 1.53 & 0.28 & 2.31 & 1.23 \\
    \hline
    gplus(S) & GP & 0.11 & 12.23 & 218.2 & 20,127  & 5.57 & 0.80 & 46.5 & 10.68 \\
    \hline
    usa-road(D) & US & 23.95 & 28.85 & 2.41 & 9  & 12.43 & 1.89 & 3.69 & 1.90 \\
    \hline
    livejournal(S) & LJ & 4.85 & 43.37 & 17.88 & 20,333  & 14.25 & 2.81 & 20.33 & 12.49\\
    \hline
    uk2002(W) & UK & 18.50 & 298.11 & 32.23 & 194,955  & 61.99 & 17.16 & 266.60 & 156.05 \\
    \hline
    eu-road(D) & EU & 173.80 & 342.70 & 3.94 & 20 & 72.96 & 22.47 & 16.98 & 22.98 \\
    \hline
    friendster(S) & FS & 65.61 & 1806.07 & 55.05 & 5,214  & 378.26 & 118.40 & 368.95 & 395.31 \\
    \hline
    Synthetic & SY & 372.00 & 10,000.00 & 53 & 613,461  & 2027 & 493.75 & 5604.00 & 660.61 \\
    \hline
    \end{tabular}
\end{table*}

\stitle{Queries.} The queries are presented in \reffig{queries}. We also give the partial order under each query for symmetry breaking. The queries except $q_7$ and $q_8$ are selected based on all prior works \cite{Ammar2018, Lai2015, Lai2016, Qiao2017, Shao2014}, while varying in number of vertices, densities, and the vertex cover ratio $|V^{cc}_Q|/|V_Q|$, in order to better evaluate the strategies from different perspectives. The three queries $q_7$, $q_8$ and $q_9$ are relatively challenging given their result scale. For example, the smallest dataset GO contains $2,168$B(illion) $q_7$, $330$B $q_8$ and $1,883$B $q_9$, respectively. For short of space, we record the number of results of each successful query on each dataset in the appendix. Note that $q_7$ and $q_8$ are absent from existing works, while we benchmark $q_7$ considering the importance of path query in practice, and $q_8$ considering the varieties of the join plans. 

\begin{figure}[htb]
    \centering
    \includegraphics[scale=0.9]{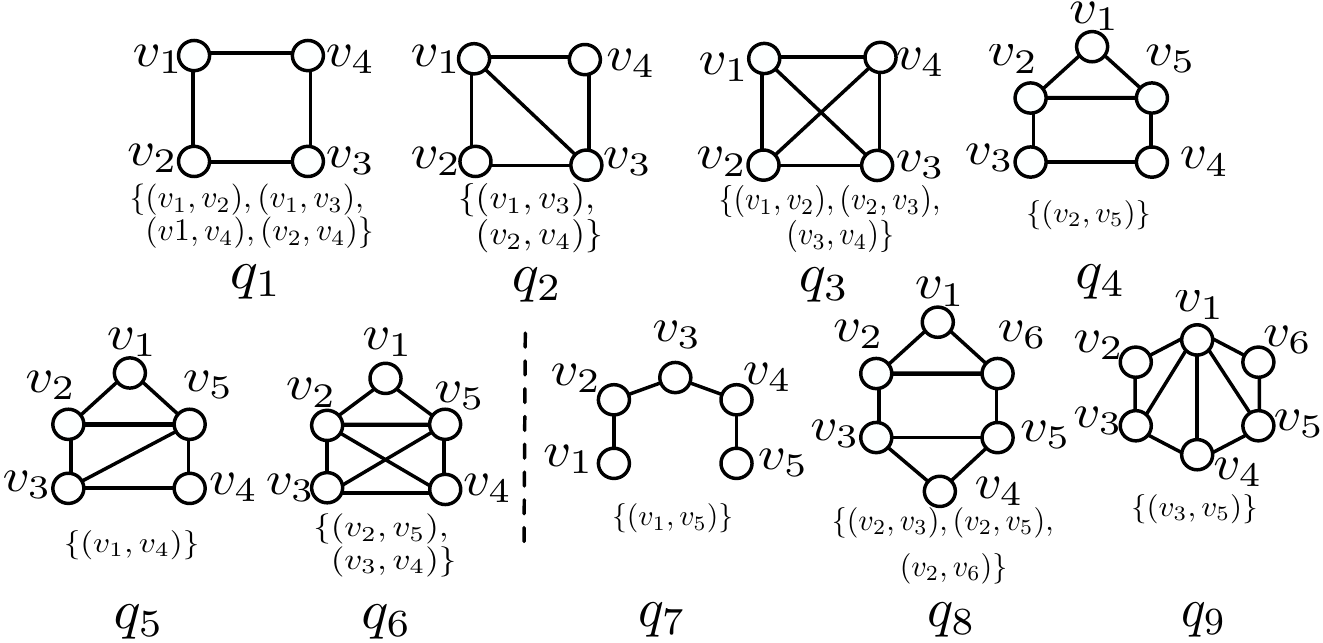}
    \caption{\small{The unlabelled queries.}}
    \label{fig:queries}
  \end{figure}

\stitle{Exp-1: Optimizations.} We study the effectiveness of \batching, \trindex and \compression for both \eaat and \vaat strategies, by comparing \eaat and \vaat with their respective variants with one optimization off, namely ``without Batching'', ``without Trindexing'' and ``without Compression''. In the following, we use the suffix of ``(w.o.b.)'', ``(w.o.t.)'' and ``(w.o.c.)'' to represent the three variants. We use the queries $q_2$ and $q_5$, and the results of US and LJ are shown in \reffig{vary_opt}. By default, we use the batch size of $1,000,000$ for both \eaat and \vaat (according to \cite{Ammar2018}) in this experiment, and we reduce the batch size when it runs out of memory, as will be specified.

\begin{figure*}[htb]
    \centering
    \begin{subfigure}[b]{\textwidth}
        \centering
        \includegraphics[height = 0.2in]{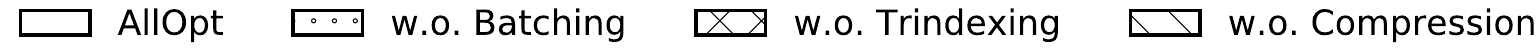}
    \end{subfigure}%
    \\
    \begin{subfigure}[b]{0.48\textwidth}
        \includegraphics[height=1.6in]{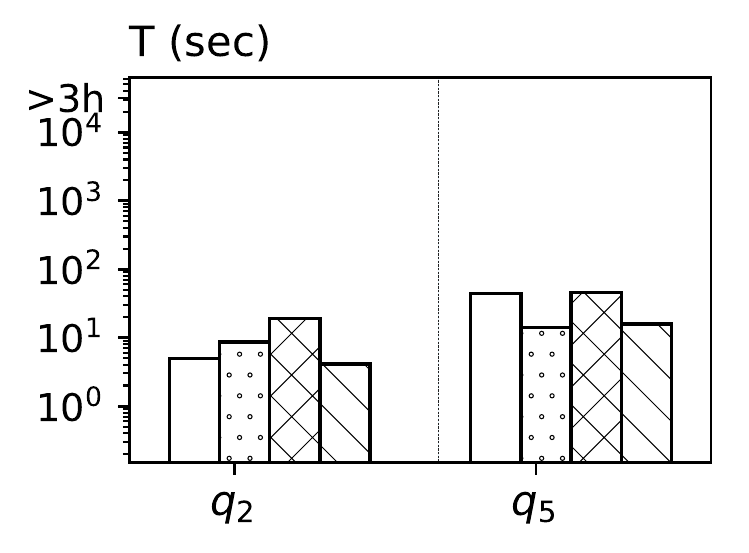}
        \caption{\eaat on US}
        \label{fig:vary_opt_eaat_us}
    \end{subfigure}%
    ~
    \begin{subfigure}[b]{0.48\textwidth}
        \includegraphics[height=1.6in]{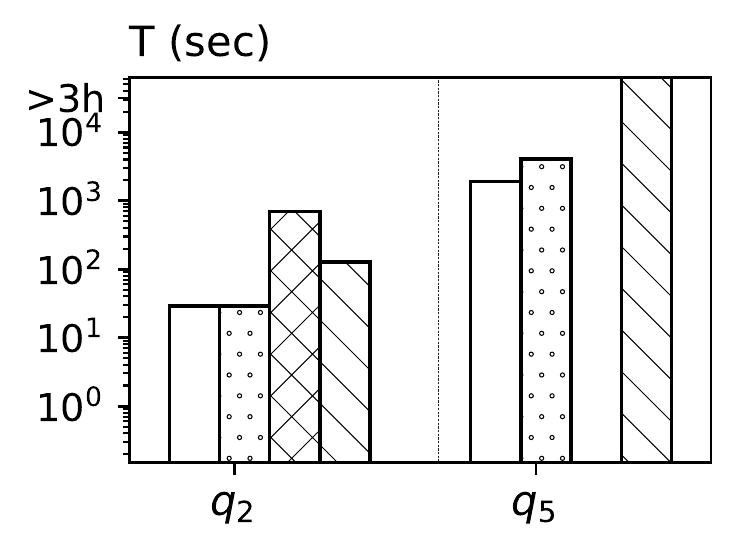}
        \caption{\eaat on LJ}
        \label{fig:vary_opt_eaat_lj}
    \end{subfigure}%
    \\
    \begin{subfigure}[b]{0.48\textwidth}
        \includegraphics[height=1.6in]{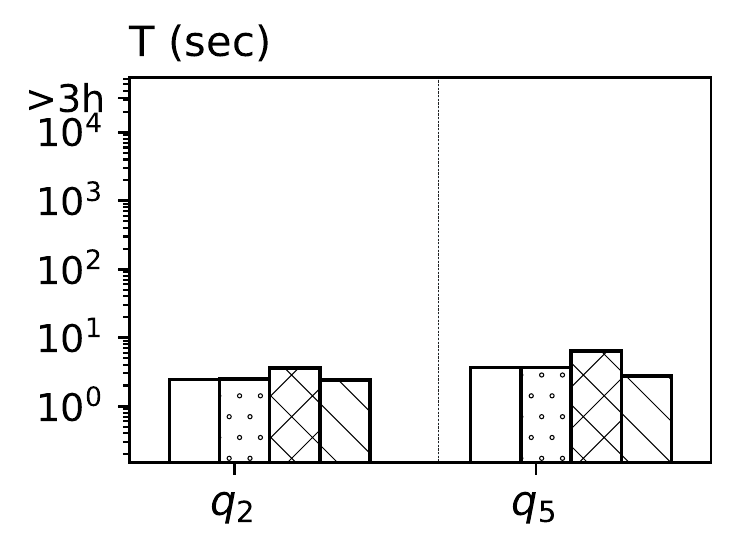}
        \caption{\vaat on US}
        \label{fig:vary_opt_vaat_us}
    \end{subfigure}%
    ~
    \begin{subfigure}[b]{0.48\textwidth}
        \includegraphics[height=1.6in]{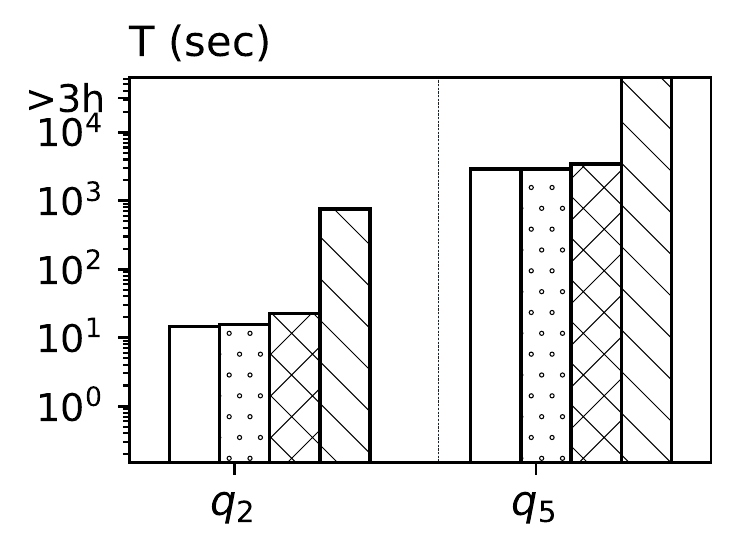}
        \caption{\vaat on LJ}
        \label{fig:vary_opt_vaat_lj}
    \end{subfigure}%
    \caption{Effectiveness of optimizations.}
    \label{fig:vary_opt}
\end{figure*}

While comparing \eaat with \binaryjoinnobatching, we observe that \batching barely affects the performance of $q_2$, but severely for $q_5$ on LJ (1800s vs 4000s (w.o.b.) ). The reason is that we still apply join-level \batching for \binaryjoinnobatching that dumps the intermediate data to the disk (\refsec{batching}). While $q_5$'s intermediate data includes the massive results of sub-query $Q(\{v_2, v_3, v_4, v_5\})$, which incurs huge amount of disk I/O (US does not have this problem as it produces very few results). We also run $q_5$ without the join-level \batching on LJ, but it fails with \oom. For \eaat, \trindex is a critical optimization, with the observed performance of \eaat better than that of \binaryjoinnotrindex, especially so on LJ. This is expected as \binaryjoinnotrindex actually degenerates to \starjoin. \compression, on the one hand, allows \eaat to run much faster than \binaryjoinnocompression for both queries on LJ, on the other hand, makes it slower on US. The reason is that US is a sparse dataset with few room for \compression, while \compression itself incurs extra cost. We also compare \eaat with \binaryjoinnocompression on the other sparse graph EU, and the results are the same.  

For \vaat strategy, \batching has little impact to the performance. Surprisingly, after using \trindex to \vaat, the improvement by average is only around 18\%. We do another experiment in the same cluster but using 1GBps switch, which shows \vaat is over 6 times faster than \genericjoinnotrindex for both queries on LJ. Note that \timely uses separate threads to buffer received data from the network. Given the same computing speed, a faster network allows the data to be more fully buffered and hence less wait for the following computation. Similar to \eaat, \compression greatly improves the performance while querying on LJ, but the opposite on US.

\stitle{Exp-2 Challenging Queries.} We study the challenging queries $q_7$, $q_8$ and $q_9$ in this experiment. We run this experiment using \eaat, \vaat, \multiway and \rep, and show the results of US and GO (LJ failed all cases) in \reffig{challenging_query}. Recall that we split the time into computation time and communication time (\refsec{experimental_settings}), here we plot the communication time as gray filling in each bar of \reffig{challenging_query}.

\begin{figure}[htb]
    \begin{subfigure}[b]{\textwidth}
    \includegraphics[height = 0.2in]{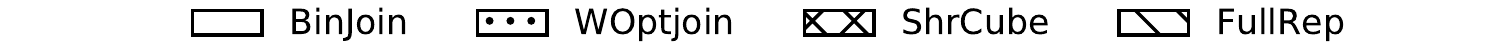}
    \end{subfigure}%
    \\
    \begin{subfigure}[b]{0.48\textwidth}
        \includegraphics[height=1.6in]{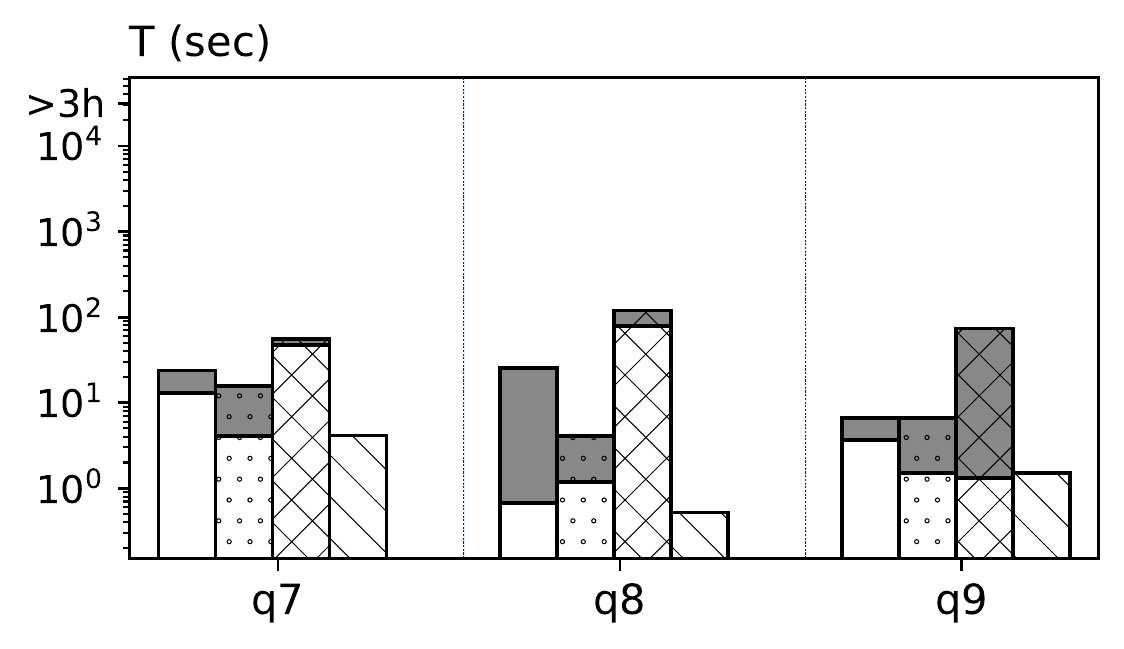}
        \caption{US}
        \label{fig:challenging_query_US}
    \end{subfigure}%
    ~
    \begin{subfigure}[b]{0.48\textwidth}
        \includegraphics[height=1.6in]{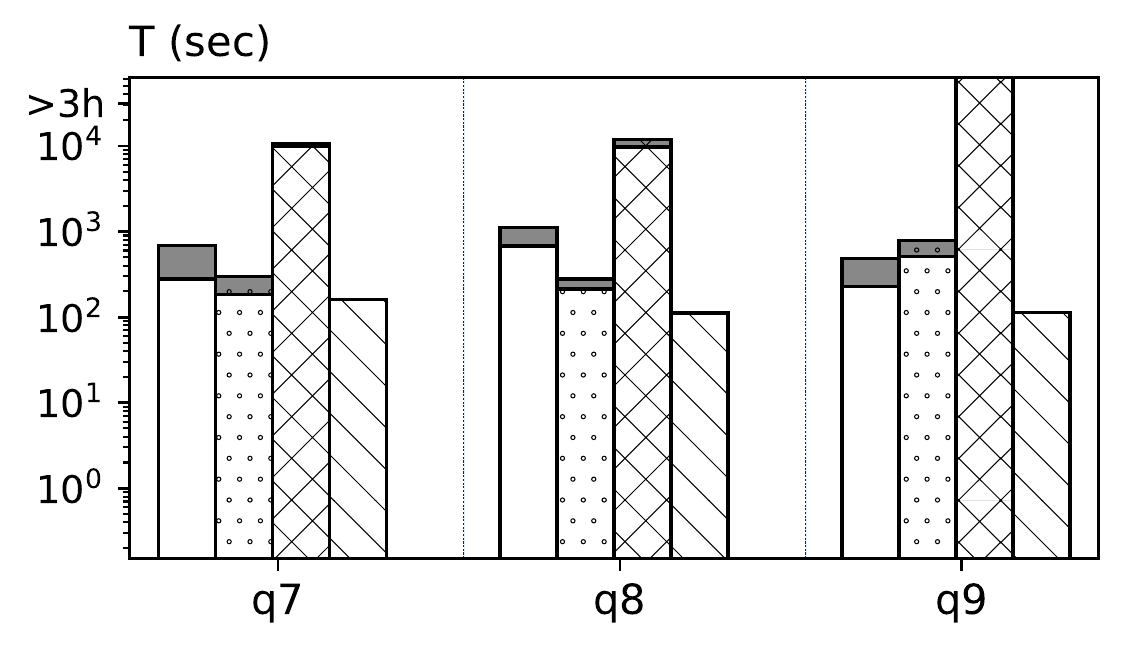}
        \caption{GO}
        \label{fig:challenging_query_GO}
    \end{subfigure}%
    \caption{Challenging queries.}
    \label{fig:challenging_query}
\end{figure}

\rep beats all the other strategies, while \multiway fails $q_8$ and $q_9$ on GO because of \timeout. Although \multiway uses the same local algorithm as \rep, it spends a lot of time on deduplication (\refsec{multiway_join}).

We focus on comparing \eaat and \vaat on GO dataset. On the one hand, \vaat outperforms \eaat for $q_7$ and $q_8$. Their join plans of $q_7$ are nearly the same except that \eaat relies on a global shuffling on $v_3$ to processing join, while \vaat sends the partial results to the machine that maintains the vertex to grow. It is hence reasonable to observe \eaat's poorer performance for $q_7$ as shuffling is typically a more costly operation. The case of $q_8$ is similar, so we do not further discuss. On the other hand, even living with costly shuffling, \eaat still performs better for $q_9$. Due to the vertex-growing nature, \vaat's ``optimal plan'' will have to process the costly sub-query $Q(\{v_1, v_2, v_3, v_4, v_5\})$. On US dataset, \vaat consistently outperforms \eaat for these queries. This is because that US does not produce massive intermediate results as LJ, thus \eaat's shuffling cost consistently dominates.

While processing complex queries like $q_8$ and $q_9$, we can study varieties of join plans for \eaat and \vaat. First of all, we want the readers to note that \eaat's join plan for $q_8$ is different from the optimal plan originally given \cite{Lai2016}. The original ``optimal'' plan computes $q_8$ by joining two tailed triangles (triangle tailed with an edge), while this alternative plan works better by joining the uppers ``house-shape'' sub-query with the bottom triangle. In theory, the tailed triangle has worse-case bound (AGM bound \cite{Ngo2018}) of $O(M^{2})$, smaller than the house's $O(M^{2.5})$, and \eaat's actually favors this plan based on cost estimation. However, we find out that the number of tailed triangles is very close to that of the houses on GO, which renders costly process for the original plan to join two tailed triangles. This indicates insufficiency of both cost estimation proposed in \cite{Lai2016} and worst-case optimal bound \cite{Ammar2018} while computing the join plan, which will be further discussed in \refsec{discussions}.

 Secondly, it is worth noting that we actually report the result of \vaat for $q_9$ while using the \crystaljoin plan, as it works better than \vaat's original ``optimal'' plan. For $q_9$, \crystaljoin will first compute $Q(V_Q^{cc})$, namely the 2-path $\{v_1, v_3, v_5\}$, thereafter it can compress all remaining vertices $v_2, v_4$ and $v_6$. In comparison, the ``optimal'' plan can only compress $v_2$ and $v_6$. In this case, \crystaljoin performs better because it configures larger compression. In \cite{Qiao2017}, the authors proved that it renders maximum compression to use the vertex cover as the uncompressed core. However, this may not necessarily result in the best performance, considering that it can be costly to compute the core part. In our experiments, the unlabelled $q_4$, $q_8$ and labelled $q_8$ are cases that \crystaljoin plan performs worse than the original \bigjoin plan (with \compression optimization), where \crystaljoin plan does not render strictly larger compression while having to process the costly core part. As a result, we only recommend \crystaljoin plan when it leads to strictly larger compression.

The final observation is that the computation time dominates most of the evaluated cases, except \eaat's $q_8$, \vaat and \multiway's $q_9$ on US. We will further discuss this in Exp-3.

\stitle{Exp-3 All-Around Comparisons.} In this experiment, we run $q_1 - q_6$ using \eaat, \vaat, \multiway and \rep across the datasets GP, LJ, UK, EU and FS. We also run \vaat with \crystaljoin plan in $q_4$ as it is the only query that renders different \crystaljoin plan from \bigjoin plan, and the results show that the performance with \bigjoin plan is consistently better. We report the results in \reffig{all_cases}, where the communication time is plotted as gray filling. As a whole, among all 35 test cases, \rep achieves the best 85\% completion rate, followed by \vaat and \eaat which complete 71.4\% and 68.6\% respectively, and \multiway performs the worst with just 8.6\% completion rate.

\begin{figure*}[htb]
    \centering
    \begin{subfigure}[b]{\textwidth}
    \centering
        \includegraphics[height = 0.2in]{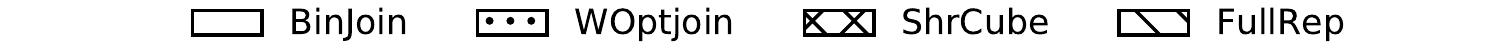}
    \end{subfigure}%
    \\
    \begin{subfigure}[b]{0.48\textwidth}
        \centering
        \includegraphics[height=1.6in]{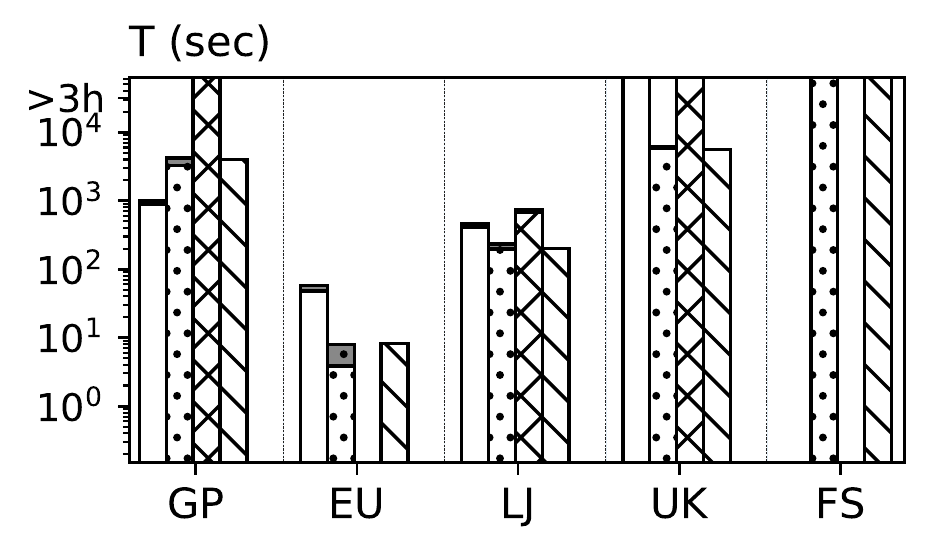}
        \caption{$q_1$}
    \end{subfigure}%
    ~
    \begin{subfigure}[b]{0.48\textwidth}
        \centering
        \includegraphics[height=1.6in]{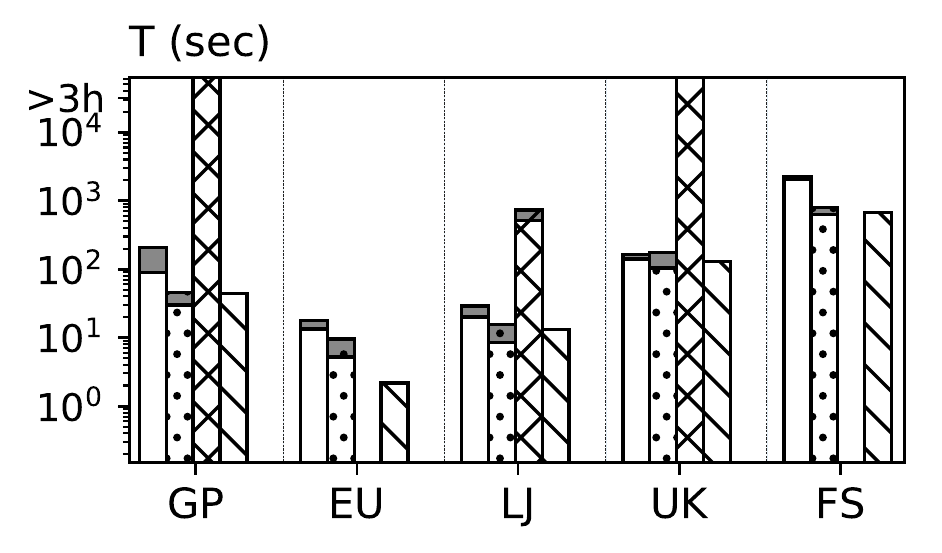}
        \caption{$q_2$}
    \end{subfigure}%
    \\
    \begin{subfigure}[b]{0.48\textwidth}
        \centering
        \includegraphics[height=1.6in]{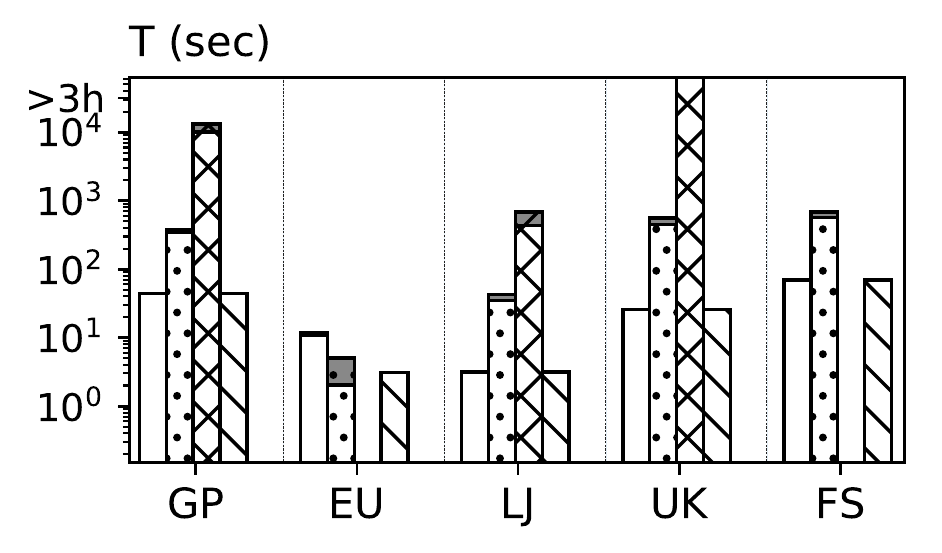}
        \caption{$q_3$}
    \end{subfigure}%
    ~
    \begin{subfigure}[b]{0.48\textwidth}
        \centering
        \includegraphics[height=1.6in]{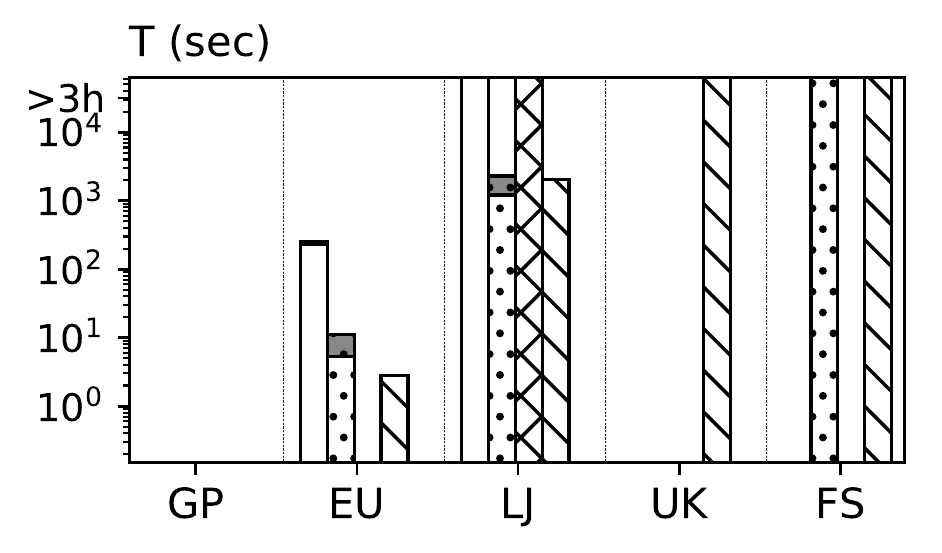}
        \caption{$q_4$}
    \end{subfigure}%
    \\
    \begin{subfigure}[b]{0.48\textwidth}
        \centering
        \includegraphics[height=1.6in]{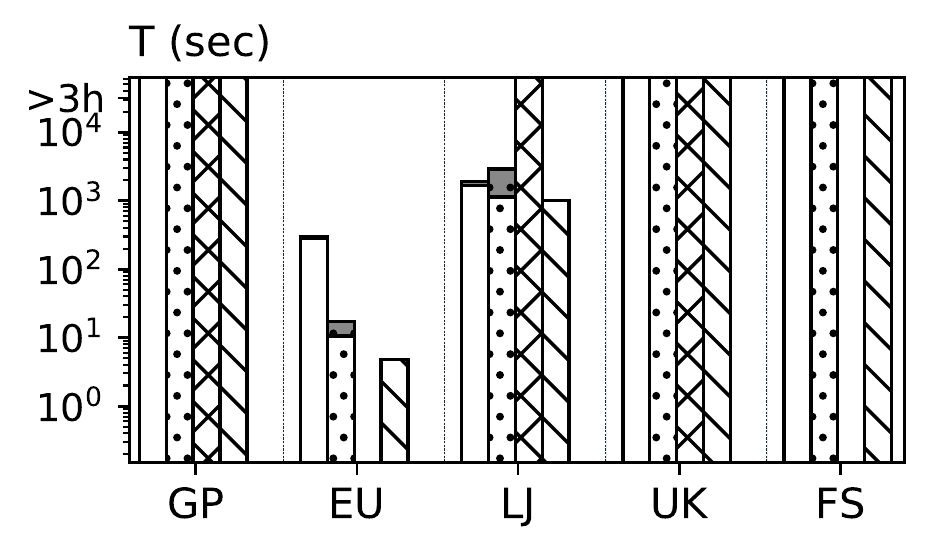}
        \caption{$q_5$}
    \end{subfigure}%
    ~
    \begin{subfigure}[b]{0.48\textwidth}
        \centering
        \includegraphics[height=1.6in]{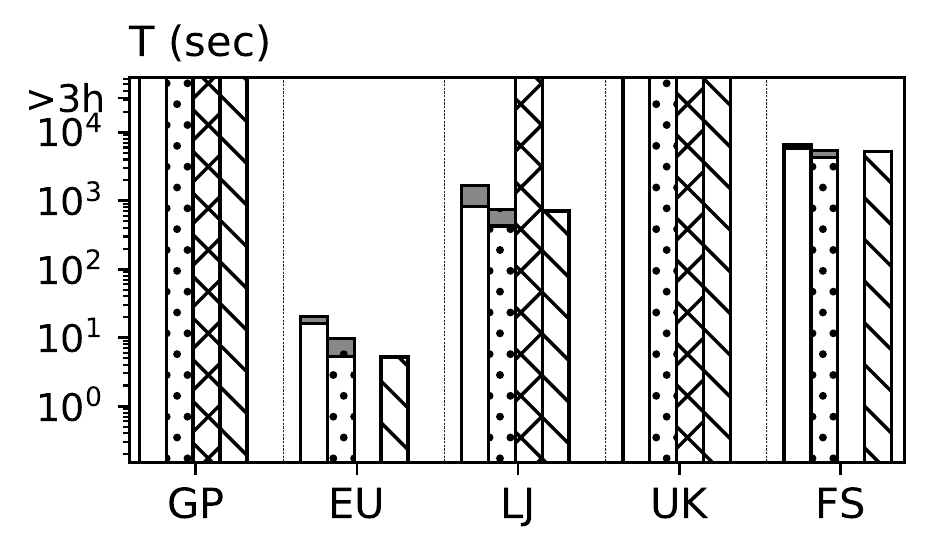}
        \caption{$q_6$}
    \end{subfigure}%
    \caption{All-around comparisons.}
    \label{fig:all_cases}
\end{figure*}

\rep typically outperforms the other strategies. Observe that \vaat's performance is often very close to \rep. The reason is that the \vaat's computing plans for these evaluated queries are similar to ``DualSim'' adopted by \rep. The extra communication cost of \vaat has been reduced to very low while adopting \trindex optimization. While comparing \vaat with \eaat, \eaat is better for $q_3$, a clique query (join unit) that requires no join (a case of embarrassingly parallel). \eaat performs worse than \vaat in most other queries, which, as we mentioned before, is due to the costly shuffling. There is an exception - querying $q_1$ on GP - where \eaat performs better than both \rep and \vaat. 
We explain this using our best speculation. GP is a very dense graph, where we observe nearly 100 vertices with degree around 10,000. To process $q_1$, after computing the sub-query $Q(\{v_1, v_2, v_4\})$, \vaat (and ``DualSim'') processes the intersection of $v_1$ and $v_4$ (their matches) for $v_3$. Those larger-degree vertices are now frequently pairing, leading to expensive intersection. In comparison, \eaat computes $q_1$ by joining the sub-query $Q(\{v_1, v_2, v_3\})$ with $Q(\{v_1, v_3, v_4\})$. Because both strategies compute $Q(\{v_1, v_2, v_3\})$, we consider how \eaat computes $Q(\{v_1, v_3, v_4\})$. \eaat first locate the matched vertex of $v_3$, then matches $v_1$ and $v_4$ among its neighbors, which is generally cheaper than intersecting the neighbors of $v_1$ and $v_4$ to compute $v_3$. 
Due to the existence of these high-degree pairs, the cost \vaat's intersection can exceed \eaat's shuffling.


We observe that the computation time $T_{comp}$ dominates in most cases as we mentioned in Exp-2. This is trivially true for \multiway and \rep, but it may not be clearly so for \vaat and \eaat given that they all need to transfer a massive amount of intermediate data. We investigate this and find out two potential reasons. The first one attributes to \timely's highly optimized communication component, which allows the computation to overlap communication by using extra threads to receive and buffer the data from the network so that it can be mostly ready for the following computation. The second one is the fast network. We re-run these queries using the 1GBps switch, while the results show the opposite trend that the communication time $T_{comm}$ in turn takes over.

\stitle{Exp-4 Web-Scale.} We run the SY datasets in the AWS cluster of 40 instances. Note that \rep can not be used as SY is larger than the machine's memory. We use the queries $q_2$ and $q_3$, and present the results of \eaat and \vaat (\multiway fails all cases due to \oom) in \reftable{web_scale_exp}. The results are consistent with the prior experiments, but observe that the gap between \eaat and \vaat while querying $q_1$ is larger. This is because that we now deploy 40 AWS instances, and \eaat's shuffling cost increases.

\begin{table}[htb]
    \centering
    \caption{The web-scale queries.}
    \label{tab:web_scale_exp}
     \begin{tabular}{|c|c|c|} 
     \hline
     Queries & \eaat ($T_{comp}$)/s & \vaat ($T_{comp}$)/s \\
     \hline\hline
     $q_2$ & 8810 (6893) & 1751 (1511) \\
    \hline
     $q_3$ & 76 (75) & 518 (443) \\   
    \hline
     \end{tabular}
\end{table}

\subsection{Labelled Experiments}
\label{sec:labelled_experiment}
We use the LDBC social network benchmarking (SNB) \cite{Ldbc} for labelled matching experiment due to the lack of labelled big graphs in the public. SNB provides a data generator that generates a synthetic social network of required statistics, and a document \cite{Ldbc_doc} that describes the benchmarking tasks, in which the complex tasks are actually subgraph matching. The join plans of \eaat and \vaat for labelled experiments are generated as unlabelled case, but we use the label frequencies to break tie.

\stitle{Datasets.} We list the datasets and their statistics in \reftable{labelled_datasets}. These datasets are generated using the "Facebook" mode with a duration of 3 years. The dataset's name, denoted as DG$x$, represents a scale factor of $x$. The labels are preprocessed into integers. 

\begin{table}[htb]
\centering
\caption{The labelled datasets.}
\label{tab:labelled_datasets}
 \begin{tabular}{|r|c|c|c|c|c|} 
 \hline
 Name & $|V_G|$ & $|E_G|$ & $\overline{d}_G$ & $D_G$ & \# Labels \\
 \hline\hline
 DG10 & 29.99 & 176.48 & 11.77 & 4,282,812 & 10 \\
\hline
 DG60 & 187.11 & 1246.66 & 13.32 & 26,639,563 & 10 \\   
\hline
 \end{tabular}
\end{table}

\begin{figure}[htb]
  \centering
  \includegraphics[scale=0.65]{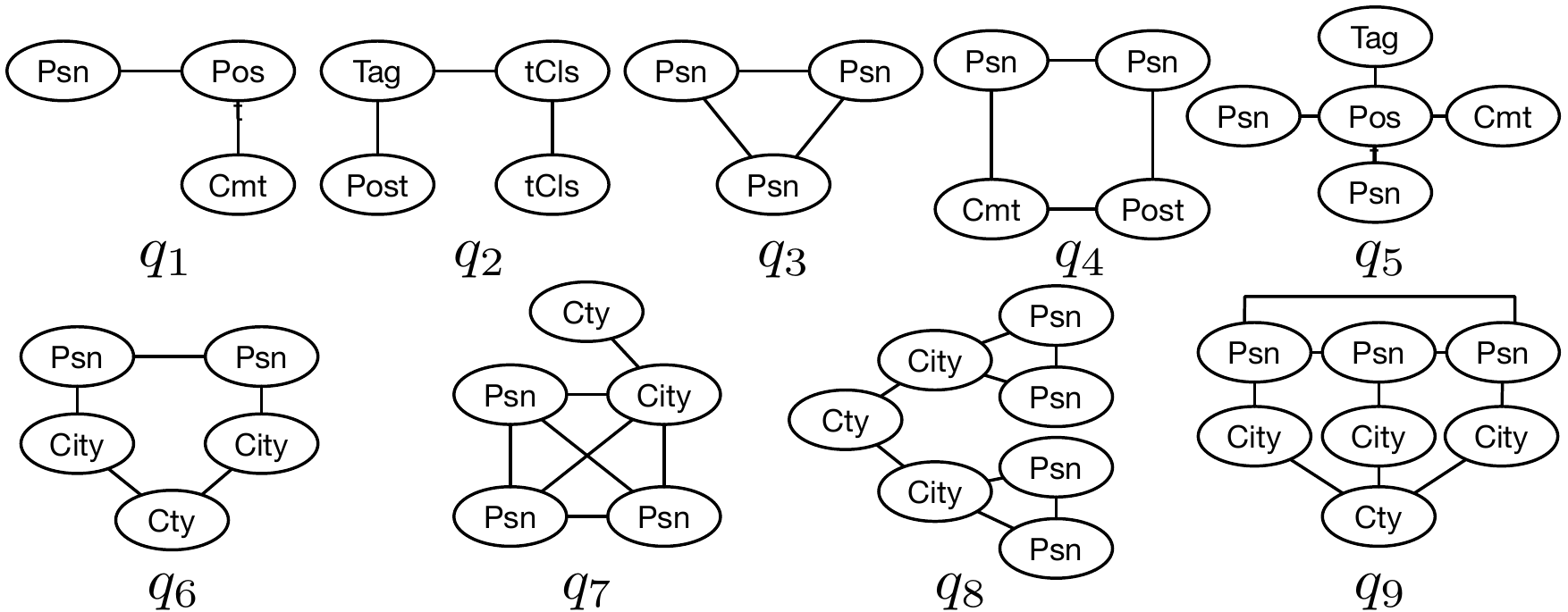}
  \caption{\small{Labelled queries.}}
  \label{fig:labelled_queries}
\end{figure}

\stitle{Queries.} The queries, shown in \reffig{labelled_queries}, are selected from the SNB's complex tasks with some adaptions: (1) removing the direction of the edges; (2) removing the edge labels; (3) using one-hop edge for multi-hop edges; (4) removing the ``no edge'' condition; (5) removing all properties except the node types as labels. For (1) and (2), note that our current implementation can support both cases, and we do the adaptations for consistency and simplicity. For (3) and (4), we adapt them because currently they do not conform with the subgraph matching problem studied in this paper. For (5), it is due to our current limitation in supporting property graph. We leave (3), (4) and (5) as interesting future work.

\stitle{Exp-5 All-Around Comparisons.}  We now conduct the experiment using all queries on DG10 and DG60, and present the results in \reffig{all_labelled_cases}. Here we compute the join plans for \eaat and \vaat by using the unlabelled method, but further using the label frequencies to break tie. The gray filling again represents communication time. \rep outperforms the other strategies in many cases, except that it performs slightly slower than \eaat for $q_3$ and $q_5$. This is because that $q_3$ and $q_5$ are join units, and \eaat processes them locally in each machine as \rep, and it does not build indices as ``CFLMatch'' used in \rep. When comparing to \vaat,   Among all these queries, we only have $q_8$ that configures different \crystaljoin plan (w.r.t. \bigjoin plan) for \vaat. The results show that the performance of \vaat drops about 10 times while using \crystaljoin plan. Note that the core part of $q_8$ is a 5-path of ``Psn-City-Cty-City-Psn'' with enormous intermediate results. As we mentioned in unlabelled experiments, it may not always be wise to first compute the vertex-cover-induced core.

We now focus on comparing \eaat and \vaat. There are three cases that intrigue us. Firstly, observe that \eaat performs much better than \vaat while querying $q_4$. The reason is high intersection cost as we discovered on GP dataset in unlabelled matching. Secondly, \eaat performs worse than \vaat in $q_7$, which again is because of \eaat's costly shuffling. The third case is $q_9$, the most complex query in the experiment. \eaat performs much better while querying $q_9$. The bad performance of \vaat comes from the long execution plan together with costly intermediate results. The two algorithms all expand the three ``Psn''s, and then grow via one of the ``City''s to ``Cty'', but \eaat approaches this using one join (a triangle $\Join$ a \ttwig), while \vaat will first expand to ``City'' then further ``Cty'', and the ``City'' expansion is the culprit of the slower run.

\begin{figure}[htb]
    \centering
    \begin{subfigure}[b]{\textwidth}
    \centering
    \includegraphics[height = 0.2in]{vary_dataset_legend.pdf}
    \end{subfigure}%
    \\
    \begin{subfigure}[b]{\textwidth}
    \centering
    \includegraphics[height=1.6in]{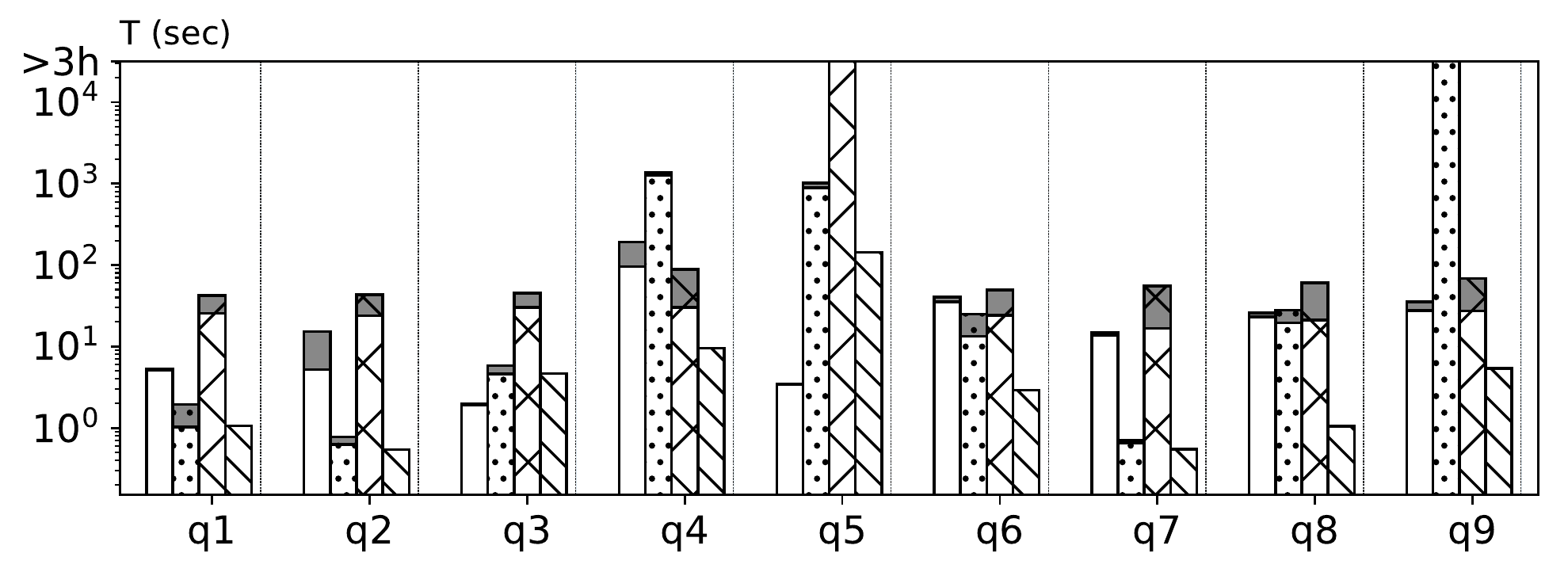}
    \caption{DG10}\label{fig:all_labelled_cases_dg10}
    \end{subfigure}%
    \\
    \begin{subfigure}[b]{\textwidth}
    \centering
    \includegraphics[height=1.6in]{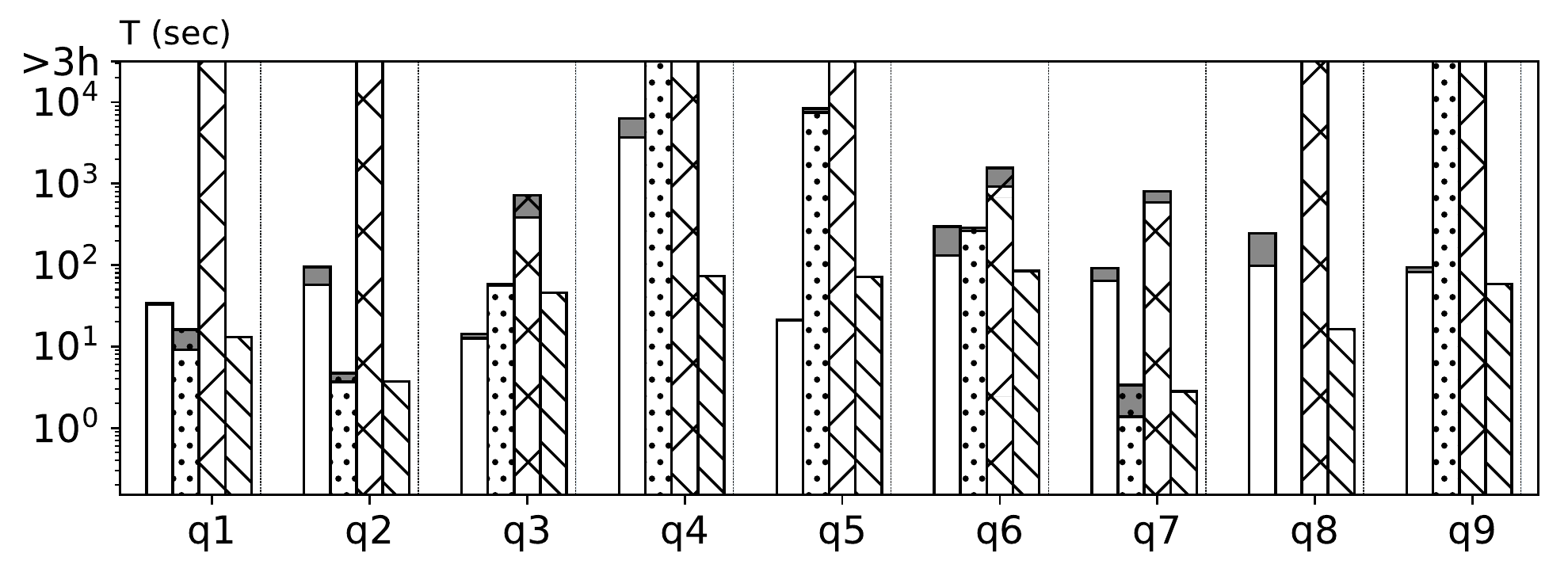}
    \caption{DG60}\label{fig:all_labelled_cases_dg60}
    \end{subfigure}%
    \caption{All-around comparisons of labelled matching.}\label{fig:all_labelled_cases}
\end{figure}

\section{Discussions and Future Work.}
\label{sec:discussions}
We discuss our findings and potential future work based on the experiments in \refsec{experiments}. Eventually, we summarize the findings into a practical guide.

\stitle{Strategy Selection.} \emph{\rep is obviously the preferred choice when the machine can hold the graph data, while both \vaat and \eaat are good alternatives when the graph is larger than the capacity of the machine.} For \eaat and \vaat, on one side, \eaat may perform worse than \vaat (e.g. unlabelled $q_2$, $q_4$, $q_5$) due to the expensive shuffling operation, on the other side, \eaat can also outperform \vaat (e.g. unlabelled and labelled $q_9$) while avoiding costly sub-queries due to query decomposition. One way to choose between \eaat and \vaat is to compare the cost of their respective join plans, and select the one with less cost. For now, we can either use cost estimation proposed in \cite{Lai2016}, or summing the worst-case bound, but none of them consistently gives the best solution, as will be discussed in ``Optimal Join Plan''. Alternatively, we refer to ``EmptyHeaded'' \cite{Aberger2016} to study a potential hybrid strategy of \eaat and \vaat. Note that ``EmptyHeaded'' is developed in single-machine setting, and it does not take into consideration the impact of \compression, we hence leave such hybrid strategy in the distributed context as an interesting future work. 

\stitle{Optimizations.} \emph{Our experimental results suggest always using \batching, using \trindex when each machine has \textbf{sufficient} memory to hold ``triangle partition'', and using \compression when the data graph is not very sparse (e.g. $\overline{d}_G \geq 5$).} \batching often does not impact performance, so we recommend always using \batching due to the unpredictability of the size of (intermediate) results. \trindex is critical for \eaat, and it can greatly improve \vaat by reducing communication cost, while it requires extra storage to maintain ``triangle partition''. Amongst the evaluated datasets, each ``triangle partition'' maintains an average of 30\% data in our 10-machine cluster. Thus, we suggest a memory threshold of $60\% |E_G|$ (half for graph and half for running algorithm) for \trindex in a cluster of the same or larger scale. Note that the threshold does not apply to extremely dense graph. Among the three optimizations, \compression is the primary performance booster that improves the performance of \eaat and \vaat by 5 times on average in all but the cases on the very sparse road networks. For such very sparse data graphs, \compression can render more cost than benefits.   


\stitle{Optimal Join Plan.} \emph{It is challenging to systematically determine the optimal join plans for both \eaat and \vaat.} From the experiments, we identify three impact factors: (1) the worst-case bound; (2) cost estimation based on data statistics; (3) favoring the optimizations, especially \compression. All existing works only partially consider these factors, and we have observed sub-optimal join plans in the experiments. For example, \eaat bases the ``optimal'' join plan on minimizing the cost estimation, but the join plan does not render the best performance for unlabelled $q_8$; \vaat follows the worst-case optimality, while it may encounter costly sub-queries for labelled and unlabelled $q_9$; \crystaljoin focuses on maximizing the compression, while ignoring the facts that the vertex-cover-induced core part itself can be costly to compute. Additionally, there are other impact factors such as the partial orders of query vertices and the label frequencies, which have not been studied in this work due to short of space. It is another very interesting future work to thoroughly study the optimal join plan while considering all above impact factors.

\stitle{Computation vs. Communication.} \emph{We argue that distributed subgraph matching nowadays is a computation-intensive task.} This claim holds when the cluster configures high-speed network (e.g. $\geq 10$GBps), and the data processor can efficiently overlap computation with communication. Note that computation cost (either \eaat's join or \vaat's intersection) is lower-bounded by the output size that is equal to the communication cost. Therefore, computation becomes the bottleneck if the network condition is good to guarantee the data to be delivered in time. Nowadays, the bandwidth of local cluster commonly exceeds 10GBps, and the overlapping of computation and communication is widely used in distributed systems (e.g. Spark \cite{Zaharia2010}, Flink \cite{Carbone2015}). 
As a result, we tend to see distributed subgraph matching as a computation-intensive task, and we advocate future research to devote more efforts into optimizing the computation while considering the following perspectives: (1) the new advancements of hardware, for example the co-processing on GPU in the coupled CPU-GPU architectures \cite{He2014} and the SIMD programming model on modern CPU \cite{Inoue2014}; (2) general computing optimizations such as load balancing strategy and cache-aware graph data accessing \cite{Wei2016}. 

\stitle{A Practical Guide.} Based on the experimental findings, we propose a practical guide for distributed subgraph matching in \reffig{program_guide}. Note that this program guide is based on current progress of the literature, and future work is needed, for examples to study the hybrid strategy and the impact factors of the optimal join plan, before we can arrive at a solid decision-making to choose between \eaat and \vaat.

\begin{figure}[htb]
    \centering
    \includegraphics[scale=0.75]{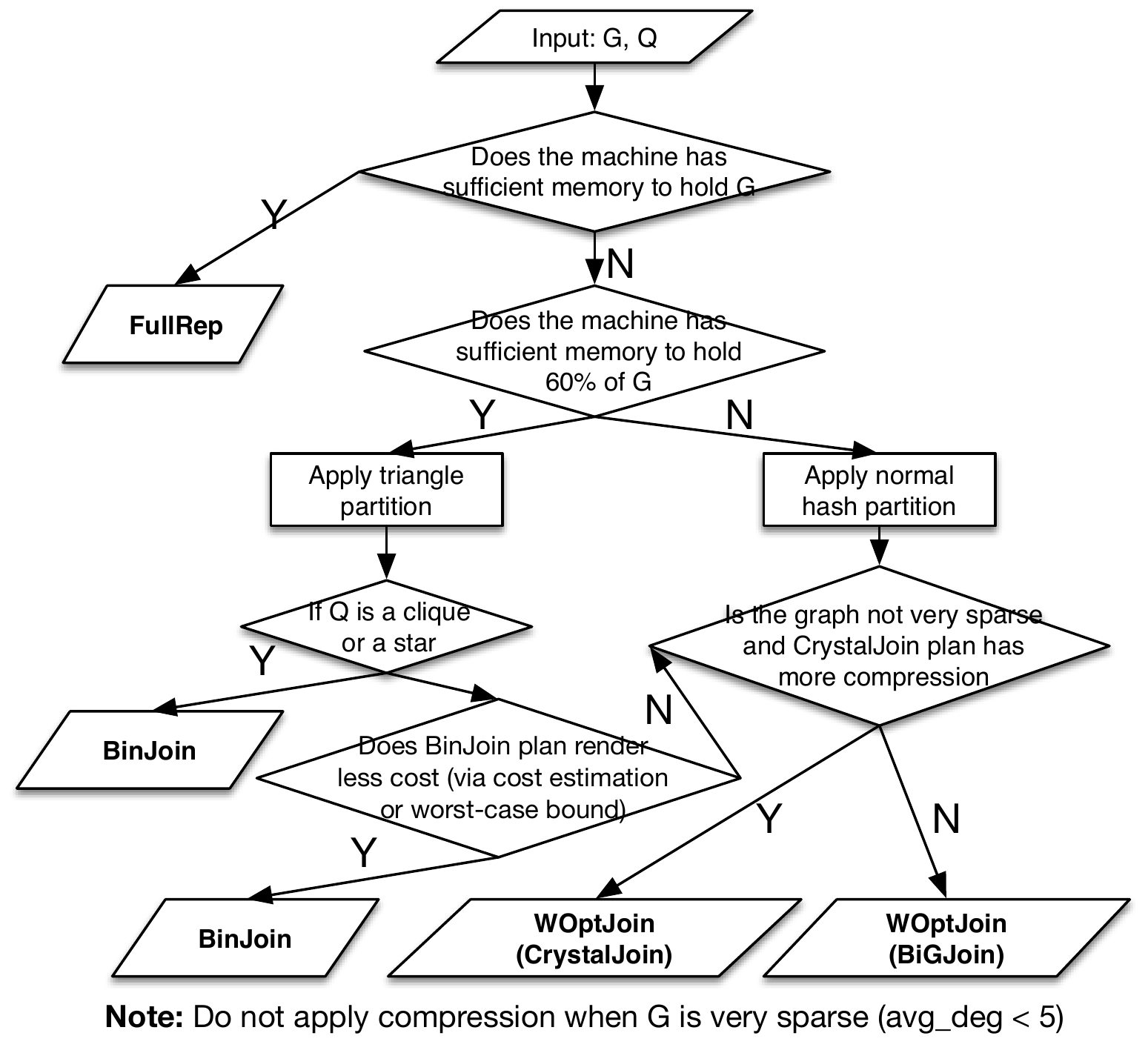}
    \caption{\small{A practical guide of distributed subgraph matching.}}
    \label{fig:program_guide}
  \end{figure}
  
\section{Related Work}
\label{sec:related_works}
\stitle{Isomorphism-based Subgraph Matching.} In the labelled case, Shang et al. \cite{Shang2008} used the spanning tree of the query graph to filter infeasible results. Han et al. \cite{Han2013} observed the importance of matching order. In \cite{Ren2015}, the authors proposed to utilize the symmetry properties in the data graph to compress the results. Bi et al. \cite{Bi2016} proposed \cflmatch based on the ``core-forest-leaves'' matching order, and obtained performance gain by postponing the notorious cartesian product. 

The unlabelled case is also known as subgraph listing/enumeration, and due to the gigantic (intermediate) results, people have been either seeking scalable algorithms in parallel, or devising techniques to compress the results. Other than the algorithms studied in this paper (\refsec{algorithms}), Kim et al. proposed the external-memory-based parallel algorithm {\scshape DualSim} \cite{Kim2016}, which maintains the data graph  in blocks on the disk, and matches the query graph by swapping in/out blocks of data to improve I/O efficiency. 

\stitle{Incremental Subgraph Matching.} Computing subgraph matching in a continuous context has recently drawn a lot of attentions. Fan et al. \cite{Fan2011} proposed incremental algorithm that identifies a portion of the data graph affected by the update regarding the query. The authors in \cite{Choudhury2015} used the join scheme as \eaat algorithms (\refsec{edge-at-a-time}). The algorithm maintained a left-deep join tree for the query, with each vertex maintaining a partial query and the corresponding partial results. Then one can compute the incremental answers of each partial query in response to the update, and utilizes the join tree to re-construct the results. Graphflow \cite{Kankanamge2017} solved incremental subgraph matching using join, in the sense that the incremental query can be transformed into $m$ independent joins, where $m$ is the number of query edges. Then they used the worst-case-optimal join algorithm to solve these joins in parallel. Most recently, Kim et al. proposed {\scshape TurboFlux} that maintains data-centric index for incremental queries, which achieves good tradeoff between performance and storage.

\stitle{Query Languages and Systems.} As the increasing demand of subgraph matching in graph analysis, people start to investigate easy-use and highly expressive subgraph matching language. Neo4j introduced \textit{Cypher} \cite{neo4j}, and now people are working on standardizing the semantics of subgraph matching based on Cypher \cite{Angles2017}. Gradoop \cite{Junghanns2017} is a system based on Apache Hadoop that translates a Cypher query into a MapReduce job. Aberger et al. proposed \textsc{EmptyHeaded} based on relational semantics for graph processing, in which they leveraged worst-case optimal join algorithm to solve subgraph matching. Arabesque \cite{Teixeira2015} was designed to solve graph mining (continuously computing frequent subgraphs) at scale, while it can be configured for single subgraph query.

\section{Conclusions}
\label{sec:conclusion}
In this paper, we implement four strategies and three general-purpose optimizations for distributed subgraph matching based on \timely dataflow system, aiming for a systematic, strategy-level comparisons among the state-of-the-art algorithms. Based on thorough empirical analysis, we summarize a practical guide, and we also motivate interesting future work for distributed subgraph matching. 


%
%
\bibliographystyle{abbrv}
\bibliography{vldb}  

\appendix
\section{Auxiliary Experiments}
\stitle{Exp-6 Scalability of Unlabelled Matching.} We vary the number of machines as 1, 2, 4, 6, 8, 10, and run the unlabelled queries $q_1$ and $q_2$ to see how each strategy (\eaat, \vaat, \multiway and \rep) scales out. We further evaluate ``Single Thread'', a serial algorithm that is specially implemented for these two queries. According to \cite{McSherry2015}, we define COST of a strategy as the number of workers it needs to outperform the ``Single Thread'', which is a comprehensive measurement of both efficiency and scalability. In this experiment, we query $q_1$ and $q_2$ on the popular dataset LJ, and show the results in \reffig{scalability}. Note that we only plot the communication and memory consumption for $q_1$, as $q_2$ follows similar trend. We also test on the other datasets, such as the dense dataset GP, the results are also similar.

\begin{figure*}[htb]
    \centering
    \begin{subfigure}[b]{\textwidth}
    \centering
        \includegraphics[height = 0.2in]{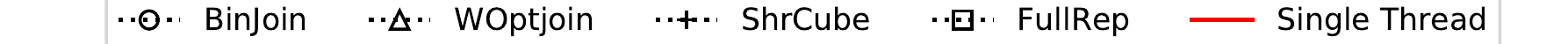}
    \end{subfigure}%
    \\
    \begin{subfigure}[b]{0.48\textwidth}
        \centering
        \includegraphics[height=1.6in]{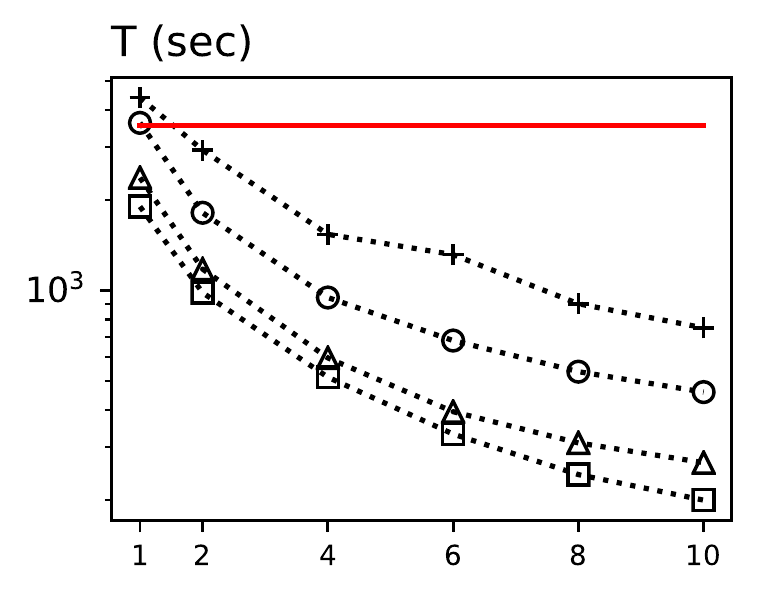}
        \caption{$q_1$ running time}
        \label{fig:scalability_time_q1}
    \end{subfigure}%
    ~
    \begin{subfigure}[b]{0.48\textwidth}
        \centering
        \includegraphics[height=1.6in]{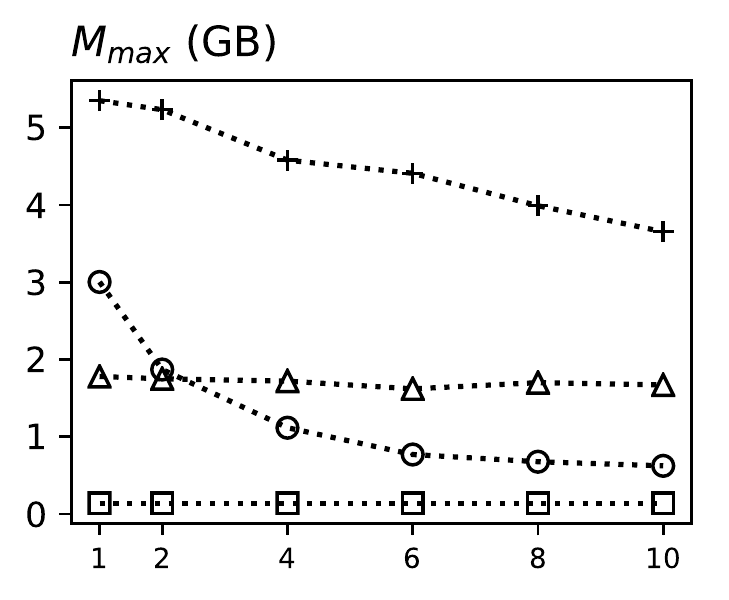}
        \caption{$q_1$ memory usage}
        \label{fig:scalability_mem_q1}
    \end{subfigure}%
    \\
    \begin{subfigure}[b]{0.48\textwidth}
        \centering
        \includegraphics[height=1.6in]{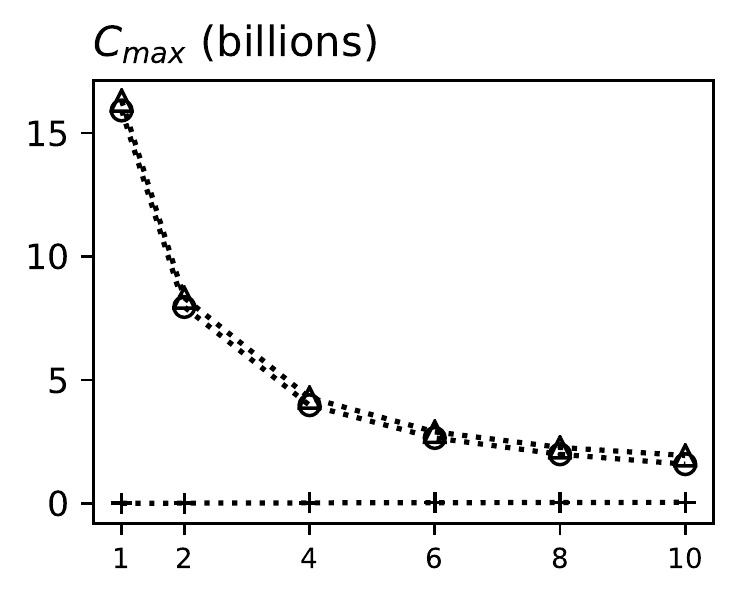}
        \caption{$q_1$ communication cost}
        \label{fig:scalability_comm_q1}
    \end{subfigure}%
    ~
	\begin{subfigure}[b]{0.48\textwidth}
        \centering
        \includegraphics[height=1.6in]{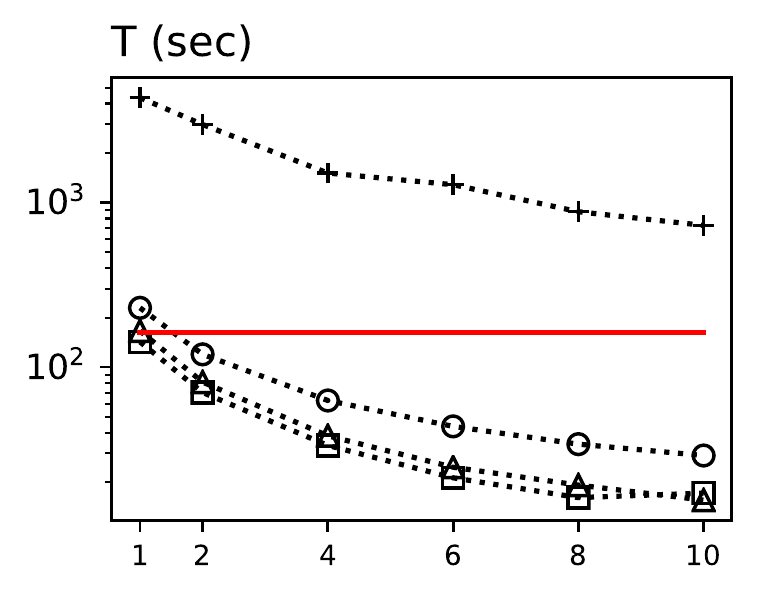}
        \caption{$q_2$ running time}
        \label{fig:scalability_time_q2}
    \end{subfigure}%
    \\
    \begin{subfigure}[b]{0.48\textwidth}
        \centering
        \includegraphics[height=1.6in]{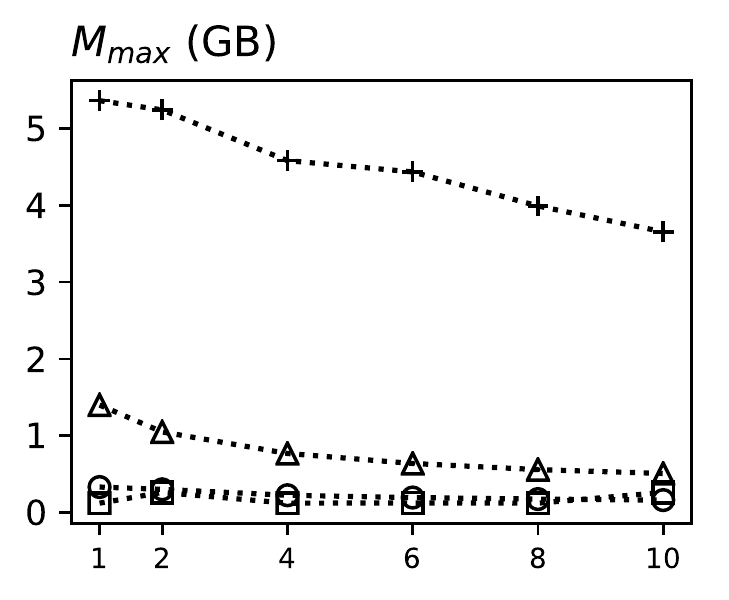}
        \caption{$q_2$ memory usage}
        \label{fig:scalability_mem_q2}
    \end{subfigure}%
    ~
    \begin{subfigure}[b]{0.48\textwidth}
        \centering
        \includegraphics[height=1.6in]{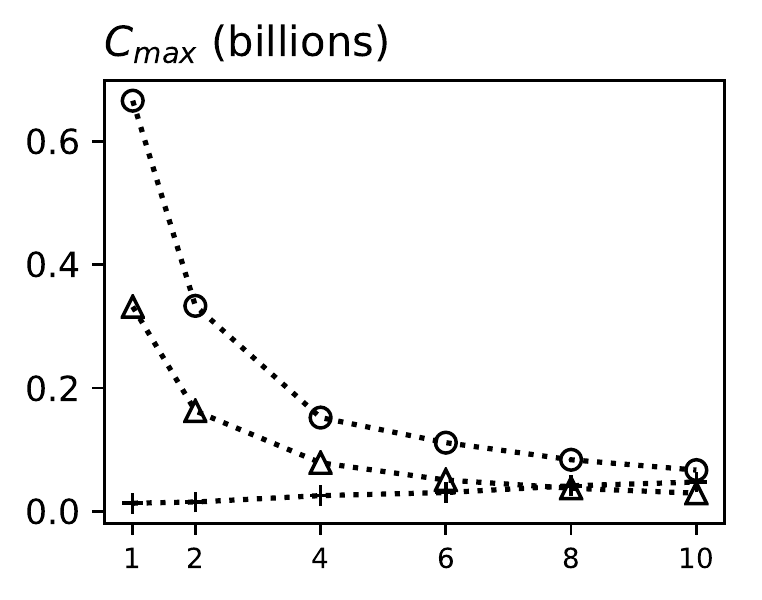}
        \caption{$q_2$ communication cost}
        \label{fig:scalability_comm_q2}
    \end{subfigure}%
    \caption{Scalability experiment: querying $q_1$ and $q_2$ on LJ.}
    \label{fig:scalability}
\end{figure*} 

All strategies demonstrate reasonable scaling regarding both queries. In terms of COST, note that \rep is slightly larger than 1, because ``DualSim'' is implemented in general for arbitrary query, while ``SingleThread'' uses a hand-tuned implementation. We first analyze the results of $q_1$. The COST ranking is \rep (1.6), \vaat (2.0), \eaat (3.1) and \multiway (3.7). As expected, \vaat scales worse than \rep, while \eaat scales worse than \vaat because shuffling cost is increasing with the number of machines. In terms of memory consumption, it is trivial that \rep constantly consumes memory of graph size. Due to the use of \batching, both \eaat and \vaat consume very low memory for both queries. Observe that \multiway consumes much more memory than \vaat and \eaat, even more than the graph data itself. This is because that certain worker may receive more edges (with duplicates) than the graph itself, which increases the peak memory consumption. For communication cost, both \eaat and \vaat demonstrate reasonable drops as the increment of machines. \multiway renders much less communication as expected, but it shows increasing trend. This is actually a reasonable behavior of \multiway, as more machines also means more data duplicates. For $q_2$, the COST ranking is \rep (2.4), \vaat (2.75), \eaat (3.82) and \multiway (71.2). Here, \multiway is dramatically larger, with most time spending on deduplication (\refsec{multiway_join}). The trend of memory consumption and communication cost of $q_2$ is similar with that of $q_1$, thus is not further discussed.

\stitle{Exp-7 Vary Desities for Labelled Matching.} Based on DG10, We generate the datasets with densities 10, 20, 40, 80 and 160 by randomly adding edges into DG10. Note that the density-10 dataset is the original DG10 in \reftable{labelled_datasets}. We use the labelled queries $q_4$ and $q_7$ in this experiment, and show the results in \reffig{vary_density}. 

\begin{figure}[htb]
    \centering
    \begin{subfigure}[b]{\textwidth}
    \includegraphics[height = 0.2in]{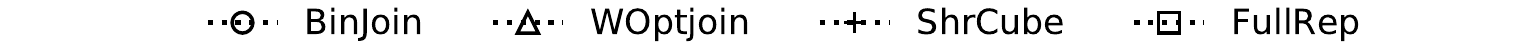}
    \end{subfigure}%
    \\
    \begin{subfigure}[b]{0.48\textwidth}
        \centering
        \includegraphics[height=1.6in]{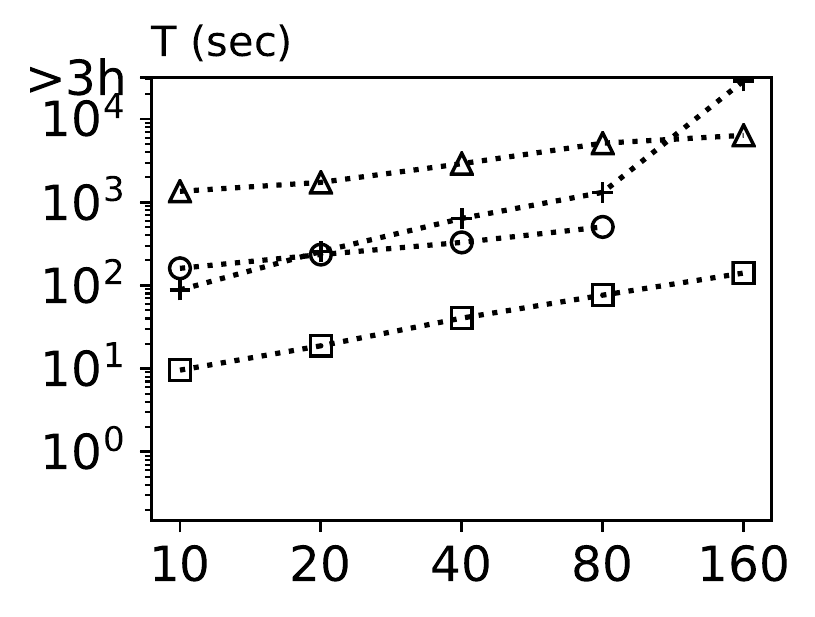}
        \caption{$q_4$}
        \label{fig:vary_density_q4}
    \end{subfigure}%
    ~
    \begin{subfigure}[b]{0.48\textwidth}
        \centering
        \includegraphics[height=1.6in]{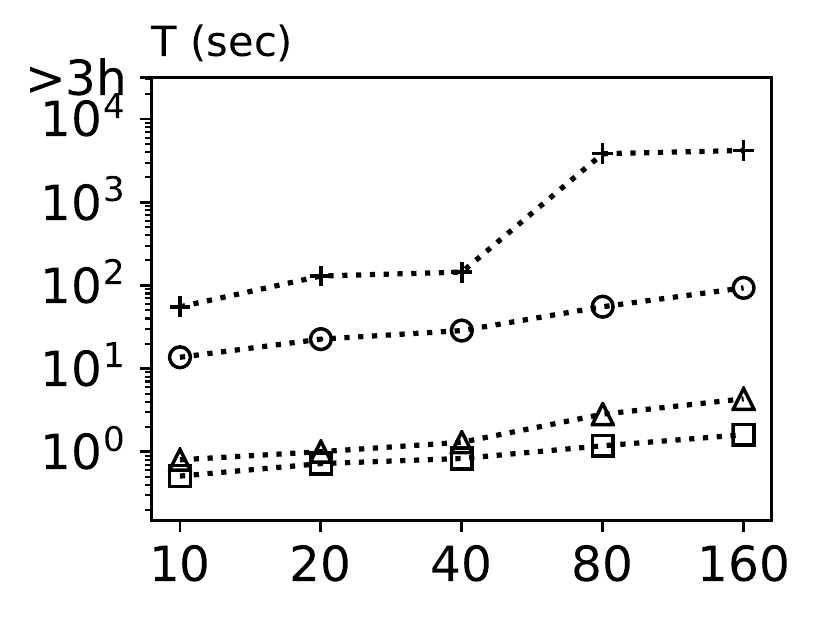}
        \caption{$q_7$}
        \label{fig:vary_density_q7}
    \end{subfigure}%
    \caption{Varying densities of labelled graph.}
    \label{fig:vary_density}
\end{figure}

\stitle{Exp-8 Vary Labels for Labelled Matching.} We generate the datasets with number of labels 0, 5, 10, 15 and 20 based on DG10. Note that there are 5 labels in labelled queries $q_4$ and $q_7$, which are called the target labels. The 10-label dataset is the original DG10. For the one with 5 labels, we will replace each label not in the target labels as one random target label. For the ones with more than 10 labels, we randomly choose some nodes to change their labels into some other pre-defined labels until they contain the required number of labels. For the one with zero label, it degenerates into unlabelled matching, and we use unlabelled version of $q_4$ and $q_7$ instead. The experiment demonstrates the transition from unlabelled matching to labelled matching, where the biggest drop happens for all algorithms. The drops continue with the increment of the number of labels, but less sharply when there are sufficient number of labels ($\geq 10$). Observe that when there are very few labels, for example, the 5-label case of $q_7$, \rep actually performs worse than \eaat and \vaat. The ``CFLMatch'' algorithm \cite{Bi2016} used by \rep relies heavily on label-based pruning. Fewer labels render larger candidate set and more recursive calls, resulting in performance drop of \rep. While fewer labels may enlarge the intermediate results of \eaat and \vaat, but they are relatively small in the labelled case, and does not create much burden for the 10GBps network.

\begin{figure}[htb]
    \centering
    \begin{subfigure}[b]{\textwidth}
    \includegraphics[height = 0.2in]{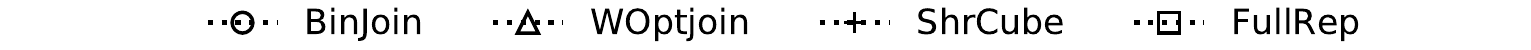}
    \end{subfigure}%
    \\
    \begin{subfigure}[b]{0.48\textwidth}
        \centering
        \includegraphics[height=1.6in]{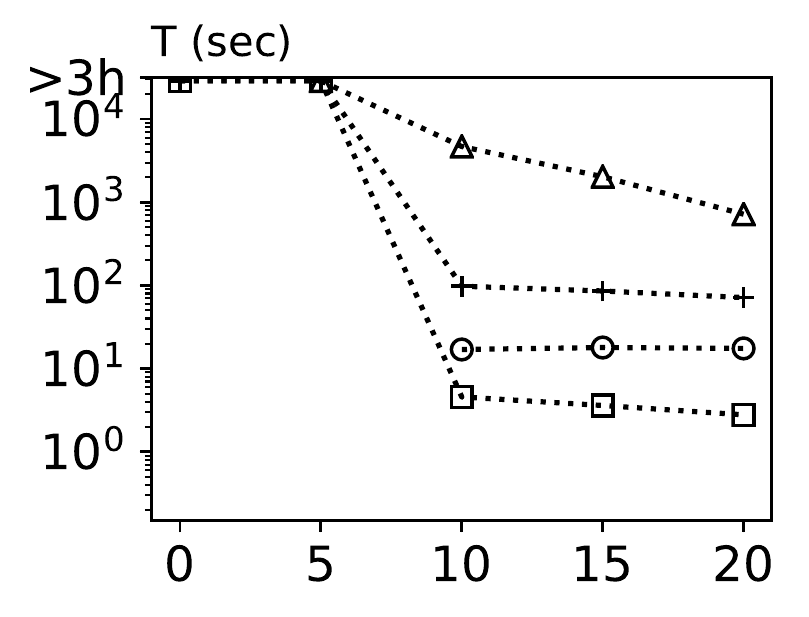}
        \caption{$q_4$}\label{fig:vary_labels_q4}
    \end{subfigure}%
    ~
    \begin{subfigure}[b]{0.48\textwidth}
        \centering
        \includegraphics[height=1.6in]{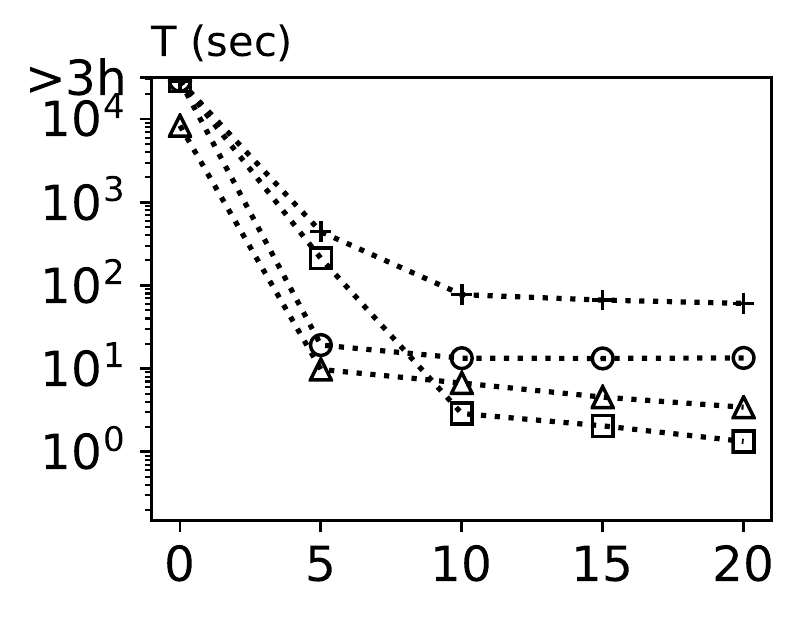}
        \caption{$q_7$}\label{fig:vary_labels_q7}
    \end{subfigure}%
    \caption{Varying labels of labelled graph.}\label{fig:vary_labels}
\end{figure}

\section{Auxiliary Materials}
\stitle{All Query Results.} In \reftable{the_results}, We show the number of results of every successful query on each dataset evaluated in this work. Note that DG10 and DG60 record the labelled queries of $q_1-q_9$.

\begin{table*}[htb]
    \centering
    \small
    \caption{All Query's Results.}
    \label{tab:the_results}
     \begin{tabular}{|c|c|c|c|c|c||c|c|c|c|} 
     \hline
     Dataset & $q_1$ & $q_2$ & $q_3$ & $q_4$ & $q_5$ & $q_6$ & $q_7$ & $q_8$ & $q_9$ \\
     \hline\hline
     GO & 539.58M & 621.18M & 39.88M & 38.20B & 27.80B & 9.28B & 2,168.86B & 330.68B & 1.88T  \\
     \hline
     GP & 1.42T & 1.16T & 78.40B & - & - & - & - & - & -  \\
     \hline
     US & 1.61M & 21,599 & 90 & 117,996 & 2,186 & 1 & 160.93M & 2,891 & 89  \\
     \hline
     LJ & 51.52B & 76.35B & 9.93B & 53.55T & 44.78T & 18.84T & - & - & - \\
     \hline
     UK & 2.49T & 2.73T & 157.19B & - & - & - & - & - & - \\
     \hline
     EU & 905,640 & 2,223 & 6 & 12,790 & 450 & 0 & 342.48M & 436 & 71 \\
     \hline
     FS & - & 185.19B & 8.96B & - & - & 3.18T & - & - & - \\
     \hline
     SY & - & 834.78B & 5.47B & - & - & - & - & - & - \\
     \hline
     \hline
     DG10* & 40.14M & 26.76M & 28.73M & 22.59M & 23.08B & 1.49M4 & 47,556 & 42.56M & 10.07M \\
    \hline
     DG60* & 302.41M & 169.86M & 267.38M & 161.69M & 203.33B & 12.44M & 983,370 & 4.14B & 114.19M \\   
    \hline
     \end{tabular}
\end{table*}  

\comment{
\stitle{Join Plans.} We show the detailed join plans of the queries we used in Exp-3 (\refsec{experiments}). We do not present $P_2$ as it is randomly generated, and gives poor performance in the experiment.

\begin{figure*}[ht]
    \begin{subfigure}[b]{0.45\textwidth}
        \centering
        \includegraphics[height=3.0in]{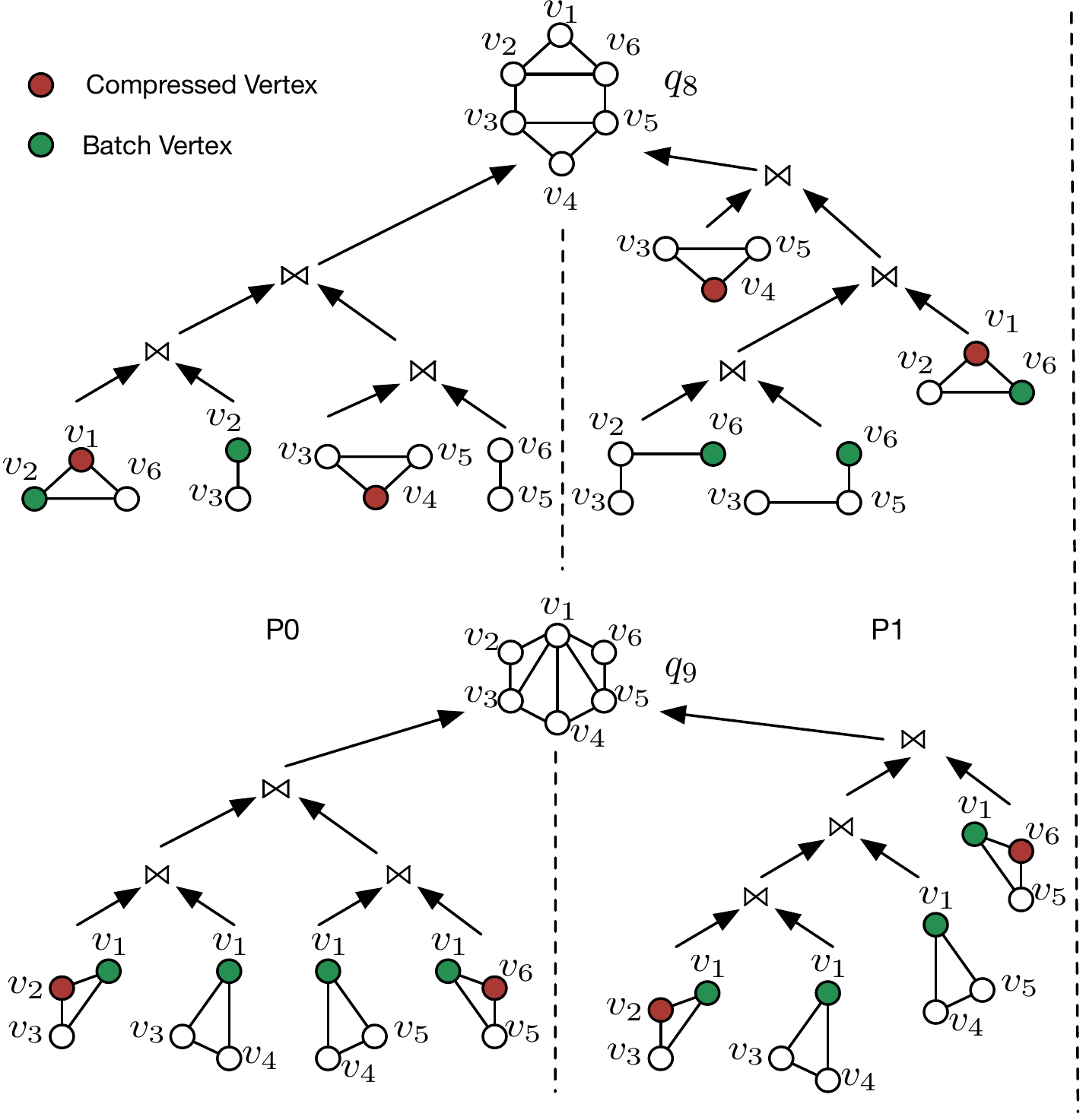}
        \caption{\eaat}
        \label{fig:binjoin_plans}
    \end{subfigure}%
    ~~
    \begin{subfigure}[b]{0.45\textwidth}
        \centering
        \includegraphics[height=3.0in]{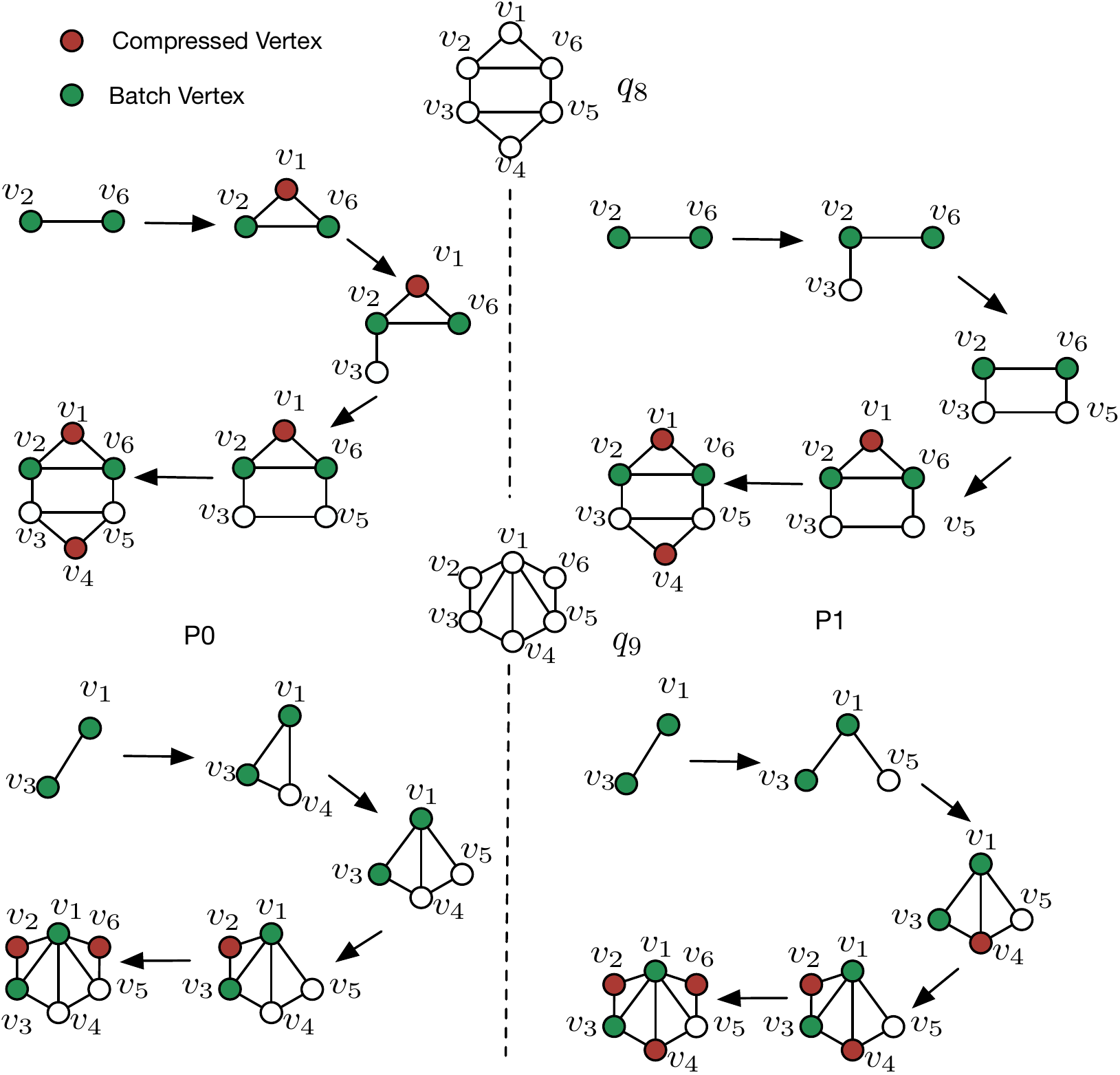}
        \caption{\vaat}
        \label{fig:woptjoin_plan}
    \end{subfigure}%
    \caption{The Join Plans of $q_8$ and $q_9$ in Exp-2.}
    \label{fig:binjoin_plan}
\end{figure*}
}

\end{document}